\documentclass[lettersize,journal]{IEEEtran}

\usepackage[caption=false,font=normalsize,labelfont=sf,textfont=sf]{subfig}

\usepackage{amsmath,amssymb,amsfonts}

\usepackage{algorithm}
\usepackage{algpseudocode}

\usepackage{graphicx}
\usepackage{mathtools}
\usepackage{color}
\usepackage{soul}
\usepackage[table]{xcolor}

\usepackage{bm}
\usepackage{amsthm}
\usepackage{booktabs}
\usepackage{float}

\usepackage{multirow}

\usepackage{tabularx}

\usepackage[normalem]{ulem}
\definecolor{lightgray}{rgb}{0.9,0.9,0.9}

\begin{document}

\title{Blockchain-enabled Data Governance for Privacy-Preserved Sharing of Confidential Data}
\author{Jingchi Zhang,~\IEEEmembership{NTU Singapore}, Anwitaman Datta,~\IEEEmembership{NTU Singapore}

\thanks{This work was supported by Ministry of Education (MOE) Singapore's Tier 2 Grant Award Number MOE-T2EP20120-0003 and Ministry of Education (MOE) Singapore's Tier 1 Research Integrity Grant 2020 Award Number 020583-00001.}}



\maketitle
\begin{abstract}
In a traditional cloud storage system, users benefit from the convenience it provides but also take the risk of certain security and privacy issues. To ensure confidentiality while maintaining data sharing capabilities, the Ciphertext-Policy Attribute-based Encryption (CP-ABE) scheme can be used to achieve fine-grained access control in cloud services. However, existing approaches are impaired by three critical concerns: illegal authorization, key disclosure, and privacy leakage. 

To address these, we propose a blockchain-based data governance system that employs blockchain technology and attribute-based encryption to prevent privacy leakage and credential misuse. First, our ABE encryption system can handle multi-authority use cases while protecting identity privacy and hiding access policy, which also protects data sharing against corrupt authorities. Second, applying the Advanced Encryption Standard (AES) for data encryption makes the whole system efficient and responsive to real-world conditions. Furthermore, the encrypted data is stored in a decentralized storage system such as IPFS, which does not rely on any centralized service provider and is, therefore, resilient against single-point failures. Third, illegal authorization activity can be readily identified through the logged on-chain data. Besides the system design, we also provide security proofs to demonstrate the robustness of the proposed system. 

\end{abstract}



\section{Introduction}\label{sec:intro}
Notwithstanding the many advantages of cloud computing, which have led to its widespread adoption and continuous growth, it also presents certain risks that prompt the exploration of alternative architectures. In particular, due to the inherent centralization of cloud services, they can become a single point of failure. This presents issues regarding service availability, censorship, and end-user privacy concerns. These challenges are further exacerbated by potential insider attacks and the service provider's own agency. For instance, Apple's decision in 2021 to roll out a Child Sexual Abuse Material detection technology by scanning images stored in its iCloud service led to numerous collateral privacy concerns and criticisms\cite{VergeArticle}. Even though the original plan was eventually abandoned mainly because of strong public backlash, the fundamental vulnerability of such centralized systems being subject to privacy violations or censorship remains.    

To address some of these issues inherent in centralized cloud storage, many encryption schemes such as AES, RSA, Proxy Re-encryption (PRE), Identity-based Encryption (IBE), and Attribute-based Encryption (ABE) have been used to secure data confidentiality \cite{sudha2012enhanced,alowolodu2013elliptic,yan2009strengthen,yang2020auth}. However, some encryption schemes may not be applicable for data sharing, which is a common use case. 

As an example, consider private electronic health records (EHR) that can be accessed by a patient (data owner, \textit{DO}) and all the hospitals (data users, \textit{DU}s) that are involved in the patient's treatments. When the patient wants to claim the cost of medical treatment, the insurance company, which is a new \textit{DU}, needs to obtain the patient's EHR to make a judgment. In a traditional public-key encryption system like RSA, the patient must re-encrypt the EHR using the public key of the insurance company. Similarly, in the IBE scheme, the patient must generate a new identity-based ciphertext for the insurance company. These access control options are limited in terms of flexibility and scalability, as the \textit{DO} may have to take additional steps on demand to ensure that the encrypted content is accessible to a new \textit{DU}. In Ciphertext-Policy Attribute-based Encryption (CP-ABE) schemes, each user is associated with a set of attributes. Successful decryption can be carried out only if a user's attribute set satisfies the access policy embedded in the ciphertext. As a result, once the message has been encrypted using the CP-ABE, it becomes available to all authorized existing and future users. Therefore, CP-ABE is a promising solution for data confidentiality and data sharing. 

Traditional CP-ABE schemes may rely on intermediary entities like a Trusted Third Party (TTP) and a Central Authority (CA) for the security and trustworthiness of data access control. For normal cloud users with limited computation resources, the \textit{DO} usually outsources the ciphertext to the TTP under a specified policy, and the \textit{DU} may also delegate the decryption task to the TTP in  \cite{green2011outsourcing,zhang2020attribute}. It also requires the CA to verify attributes across different organizations and issue private keys to every user in the data-sharing system. Therefore, when analyzing the adversarial model of these schemes, researchers tend to make the assumption that TTP or CA are fully trusted \cite{li2012scalable, zhang2020attribute}. Such a system is vulnerable to access credentials misuse that a malicious CA may purposefully issue attribute keys to some unauthorized users. It is thus crucial to design decentralized alternatives for trust mechanisms and enforce traceability throughout the access control system. 

Blockchain-based access control system with CP-ABE has the potential to be an effective solution to these problems. In the absence of a TTP to control the data, each node would share a distributed ledger that keeps track of a growing list of transactions that have been verified and confirmed by consensus mechanisms before being recorded. The integrity of transactions can be secured by hashing, Merkle trees, time stamping, and incentive mechanisms. Based on these premises, several blockchain-based access control systems have been proposed since the emergence of public blockchain systems\cite{nakamoto2008bitcoin} and the advent of Attributed-based Encryption\cite{sahai2005fuzzy}. Some efforts leverage the immutable public ledger to build a transaction-based access control system for secure data sharing\cite{maesa2017blockchain, ouaddah2016fairaccess,wang2018blockchain}, while others leverage the self-executing smart contracts to establish a smart contract-based access control system for flexibility and traceability\cite{qin2021blockchain, HEI2021Making}. 

However, just employing blockchain technology and CP-ABE encryption for an access control system is inadequate for various real-world applications. On the one hand, information is not always shared inside a single domain or organization. For example, driver's licenses and university registration information may be managed by separate entities. The same attribute authority cannot be responsible for both attribute management and key distribution. Therefore, multi-authority Attribute-based Encryption \cite{chase2007multi}, first proposed by Chase in 2007, is used to solve the access problem involving attributes belonging to various authorities. This scheme permits any number of independent authorities to distribute secret keys, which are later selected by the data owner for encrypting a message. 


Another concern is the privacy issue, which encompasses both policy-hiding and receiver privacy. Since in the classic CP-ABE schemes, an access structure specified in terms of user attributes is explicitly transmitted alongside ciphertext, whoever accesses the ciphertext is also aware of the corresponding access policy. Therefore, multi-authority CP-ABE schemes \cite{chase2007multi, lewko2011decentralizing, rouselakis2015efficient} are also unsuitable for certain use cases since access policies contain sensitive information. This calls for mechanisms to hide access policies for CP-ABE systems. Additionally, \textit{DU} needs to provide a full set of user attributes to each authority for an attribute key, inevitably compromising the key receiver's privacy. 

In pursuit of addressing these concerns, several multi-authority CP-ABE schemes that feature policy-hiding have been proposed \cite{yang2018improving, michalevsky2018decentralized, zhang2021decentralizing}. Despite these efforts, they do not completely fulfill various practical requirements. Schemes such as \cite{yang2018improving,zhang2021decentralizing} are prone to the leakage of \textit{DU}'s confidential attribute information during the key generation or encryption process. Michalevsky and Joye introduced a fully policy-hiding decentralized CP-ABE scheme \cite{michalevsky2018decentralized}, which protects attribute information attached to the access policy and even addresses the issue of receiver privacy. However, we identified a vulnerability in this scheme as it is susceptible to \textbf{rogue-key attack}. In this type of attack, a malicious \textit{AA} can register an aggregate public key using public information from other honest \textit{AA}s. This key can be used to decrypt the ciphertext despite lacking the necessary attribute keys to satisfy the policy. The attack mechanism on \cite{michalevsky2018decentralized} is elaborated in Section \ref{rk-attack}. 


As a result, the challenge of securely storing user data, enabling efficient data sharing, and managing multi-authority scenarios, while concurrently maintaining a balance of privacy, transparency and traceability constitutes a complex problem that requires innovative solutions.

\subsection{Contributions}

In this paper, we propose a multi-party CP-ABE-based storage outsourcing system that uses blockchain technology to limit access credential misuse and privacy leakage. To the best of our knowledge, this is the first practical solution for storage outsourcing to achieve fine-grained access control with user anonymity and flexibility. 

The core contributions of this work are as follows:

\begin{itemize}
    \item We discern that the decentralized ABE scheme with policy-hiding presented in \cite{michalevsky2018decentralized} inherently possesses security vulnerabilities. It is susceptible to a rogue-key attack: A malicious \textit{AA} can decrypt the ciphertext without the knowledge of the attribute keys required to satisfy the policy. It also encounters a potential risk where an adversary might infer some sensitive information from the published ciphertext, a result of poorly chosen public parameters. These vulnerabilities are thoroughly analyzed in Section \ref{rk-attack} and Section \ref{isv-risk}, respectively.

    \item To counteract the rogue-key attack and alleviate some potential risks, we modify the algorithms of \textbf{Setup} and \textbf{Auth Setup} in \cite{michalevsky2018decentralized} as described in Definition \ref{def:IPPE}. Firstly, we introduce a multi-party protocol inspired by \cite{bowe2018multi} for public key generation, which is detailed in \textbf{Trusted Setup} of Section \ref{sub_sec:trusted_setup}. Secondly, we impose a prerequisite for each \textit{AA} to prove the knowledge of published information during the process of \textbf{Auth Setup}. This is elaborated in Section \ref{sub_sec:auth_setup}. We further demonstrate that our enhanced system successfully mitigates the aforementioned security concerns, as outlined in Section \ref{rk-pos} and Section \ref{isv-pos}.

    \item In order to further decentralize our proposed system, we integrate blockchain technologies such as smart contracts and content addressing, alongside multi-authority attribute-based encryption. An overview of the system architecture is presented in Section \ref{sec:overview} and pseudo code of system contracts can be found in Section \ref{sec:app}. This hybrid approach enhances the practicality and security of the system, which makes it resilient against single-point failures and misuse of credentials. Given that transparency and traceability are inherent attributes of blockchain, a blockchain-enabled ABE system actually realizes a balanced solution for data sharing while simultaneously preserving privacy.

    \item Overall, we propose a secure, privacy-preserving data governance system based on blockchain technology and an improved decentralized CP-ABE scheme with policy-hiding. Using a combination of Attribute-based Encryption (ABE) and the Advanced Encryption Standard (AES) makes the system practical. The special ABE encryption scheme is capable of handling multi-authority use cases while protecting identity privacy and ABE's policy. The adoption of AES helps assure the confidentiality of user data, which is furthermore stored in a decentralized storage system, specifically Inter Planetary File System (IPFS), which does not rely on a central service provider and hence does not become a single point failure.

\end{itemize}

\begin{table*}[t]
\centering
\caption{Summary of Access Control System using Attribute-based Encryption (Part I)}
\label{tab:comparison_part1}
\resizebox{\textwidth}{!}{%
\begin{tabular}{l|l|l|l|l|l|l|l}
\hline
\textbf{Approach} &
  \textbf{Authority} &
  \textbf{Policy} &
  \textbf{Universe} &
  \textbf{Policy-hiding} &
  \textbf{Receiver-hiding} &
  \textbf{Access Control} &
  \textbf{Storage} \\ \hline
\cite{wang2018blockchain}   & Single   & AND  & Small & No        & No  & Smart Contract & IPFS \\ \hline
\cite{qin2021blockchain}    & Multiple & LSSS & Small & No        & No  & CSP            & CSP  \\ \hline
\cite{yang2018improving}    & Multiple & LSSS & Small & No        & Yes & CSP            & CSP  \\ \hline
\cite{nishide2008attribute} & Single   & AND  & Small & Partially & No  & CSP            & CSP  \\ \hline
\cite{gao2020trustaccess}   & Single   & AND  & Large & Fully     & No  & Smart Contract & CSP  \\ \hline
\cite{cui2018efficient}     & Single   & LSSS & Large & Partially & No  & CSP            & CSP  \\ \hline
\cite{zhang2018Security}    & Single   & LSSS & Large & Partially & No  & CSP            & CSP  \\ \hline
\cite{li2022attribute}      & Multiple & LSSS & Large & Fully     & No  & CSP            & CSP  \\ \hline
\cite{zhao2022Secure}       & Multiple & LSSS & Large & Fully     & No  & CSP            & CSP  \\ \hline
\rowcolor{lightgray} This work                        & Multiple & AND  & Small & Fully     & Yes & Smart Contract & IPFS \\ \hline
\end{tabular}%
}
\end{table*}

In Table \ref{tab:comparison_part1}, we compare and position our proposed Blockchain-enabled data governance system with existing works \cite{wang2018blockchain, qin2021blockchain, nishide2008attribute, li2022attribute, gao2020trustaccess, cui2018efficient,zhang2018Security, yang2018improving, zhao2022Secure} that are closely related to ours with regard to flexibility, scalability, privacy, and decentralization across the following assessment criteria:

\begin{enumerate}
    \item Attribute Authority: Whether the authorities involved in CP-ABE schemes are divided into single thus central authority or multi-authority. 
    
    \item Policy: Linear Secret Sharing Scheme (LSSS) which supports \textbf{AND} gate, \textbf{OR} gate, and threshold gate versus only  \textbf{AND.}
    
    \item Attribute Universe: We define the complete set of supported attributes as attribute universe and only take into account two types of universe: large universe and small universe. In large universe ABE, the attribute universe size has no effect on the size of the system's public key.  
    
    \item Privacy: There are two aspects of privacy involved in CP-ABE schemes: policy-hiding and receiver-hiding. For the policy-hiding scheme, the CP-ABE system is available in two forms: fully hidden and partially hidden. The former means that none of the attributes can be revealed from the access policies, whereas the latter refers to only hiding sensitive attribute values in the access policies. For the receiver-hiding scheme, it prevents any \textit{AA}s from learning the full set of attributes the receiver (i.e., the \textit{DU}) possesses, hence relieving the \textit{DU} from disclosing them while requesting attribute keys.
    
    \item Storage: From a technical perspective, traditional cloud service provider (CSP) and decentralized storage systems such as IPFS, Storj, and Sia, are two distinct popular solutions for data storage and sharing. CSPs may take advantage of their comprehensive control over data, but end users are exposed to the risks of a single point of failure, privacy violation, and censorship. 
    
    \item Access Control: We indicate whether access control enforcement is through a smart contract and thus logically decentralized or by a cloud service provider and thus logically centralized. 
\end{enumerate}

From Table \ref{tab:comparison_part1}, we observe that very few schemes \cite{li2022attribute, yang2018improving, zhao2022Secure} achieve fine-grained access control and support multi-authority with privacy-preserved characteristics, such as policy-hiding and receiver privacy. However, they all rely on a trusted third party (TTP) or cloud service provider (CSP) to offer centralized storage and access control management and are thus susceptible to the inherent vulnerabilities of such centralized systems in terms of privacy issues. In contrast, our proposed scheme integrates an IPFS network for decentralized storage and a smart contract for access control management in order to decentralize.

\subsection{Organization}
The rest of the paper is structured as follows: Section \ref{sec:lit} contains related work that reviews traditional Attribute-based Encryption schemes and conducts an analysis of some recent solutions for access control systems with ABE technology, while Section \ref{sec:pre} summarizes the preliminaries that the techniques developed in this paper build upon. The proposed system model and protocols are discussed in depth in Section \ref{sec:overview} and Section \ref{sec:design}. Section \ref{sec:analysis}  contains systematic security analysis. Finally, our conclusions and future plans are presented in Section \ref{sec:conclusion}.

\section{Related Work}\label{sec:lit}
Attribute-based encryption (ABE) was first introduced by Sahai and Waters in 2005 \cite{sahai2005fuzzy}. Subsequently, numerous proposals for single-authority ABE system \cite{goyal2006attribute, bethencourt2007ciphertext} have been put forth. In these systems, the data owner (\textit{DO}) encrypts data and employs a boolean formula over a set of attributes to restrict access. If the data user (\textit{DU}) possesses the secret keys, issued by a central authority (\textit{CA}) that satisfy the boolean formula attached to the ciphertext, \textit{DU} can retrieve the original data. However, these single-authority ABE systems \cite{sahai2005fuzzy, goyal2006attribute, bethencourt2007ciphertext} encounter constraints such as performance bottlenecks and key escrow issues. 

Therefore, Zhang et al proposed an enhanced CP-ABE scheme \cite{zhang2014computationally}, which alleviates the performance bottleneck issue by reducing the computation cost and ciphertext length. This work has been further explored in \cite{wang2018blockchain} to create a framework that integrates decentralized storage, smart contract, and CP-ABE techniques to achieve fine-grained access control. 

Another concern with the single-authority ABE system is key escrow, where CA issues all the attribute secret keys, thereby gaining the ability to decrypt each ciphertext generated by data owners. To address this issue, Chase and Chow introduced Multi-authority Ciphertext-policy ABE (MA-CP-ABE) \cite{chase2007multi} without the need for CA. Lewko and Waters further developed this multi-authority scheme in their work \cite{lewko2011decentralizing} allowing any authority to join or leave the system independently. Based on this work, Qin et al designed a blockchain-based multi-authority access control scheme to address performance and single-point failure issues \cite{qin2021blockchain}.

In an effort to extend the usability of CP-ABE schemes, Nishide et al. presented a desirable property, hidden access policy, in \cite{nishide2008attribute}. This approach protects sensitive attribute values while leaving attribute names public, denoted as partially-hiding. Since then, multiple enhanced schemes \cite{lai2011fully, gao2020trustaccess,cui2018efficient, zhang2018Security, li2022attribute} have been proposed. To support a wide variety of access structures, a fully secure policy-hiding CP-ABE was proposed in \cite{lai2011fully}. Gao et al used the optimized scheme of \cite{lai2011fully} to build a blockchain-based access control system in \cite{gao2020trustaccess} that achieves trustworthy access while maintaining the privacy of policy and attributes. To improve the expressiveness of the access policy, a partially hidden CP-ABE scheme under the Linear Secret Sharing Scheme (LSSS) policy was proposed in \cite{cui2018efficient}. Zhang et al proposed a privacy-aware access control system in \cite{zhang2018Security}, denoted as PASH, which supports a large universe CP-ABE scheme with partially hidden CP-ABE. We note that there is another, stricter form of policy-hiding CP-ABE, fully hiding CP-ABE, with which no information of the attributes is revealed with the access policies. Currently, fully hiding CP-ABE can only be indirectly achieved by inner product encryption (IPE) or by using threshold policies \cite{zhang2020attribute}. There are several similar approaches providing policy-hiding as well as ensuring accountability for key abuse, for example, Li's work \cite{li2022attribute} based on large universe Attribute-Based Encryption construction \cite{rouselakis2013practical} and Wu's scheme \cite{wu2019efficient} based on attribute bloom filter (ABF) \cite{dong2013private}.

Nevertheless, most of the aforementioned schemes either neglect the attribute of policy-hiding or exist as single-authority ABE systems. This gap is addressed by multi-authority CP-ABE schemes with a hidden access policy \cite{zhong2016multi,belguith2018phoabe,michalevsky2018decentralized,zhao2022Secure}. The multi-authority CP-ABE scheme featuring policy-hiding was initially introduced by Zhong et al in \cite{zhong2016multi}, and subsequently improved by Belguith in \cite{belguith2018phoabe} that significantly diminishes computational cost by delegating the decryption work to a semi-trusted cloud server. In 2018, Michalevsky and Joye \cite{michalevsky2018decentralized} put forward the first practical decentralized CP-ABE scheme with the policy-hiding property. This scheme provides a security proof in the random oracle model and supports various types of access policies, including conjunctions, disjunctions, and threshold policies. Furthermore, Michalevsky and Joye addressed the issue of receiver privacy through the use of vector commitment. However, it has its limitations, including its support for only fixed-size attributes and authorities, and the requirement for coordination among authorities during the setup phase. Most critically, we demonstrate that it is vulnerable to a rogue-key attack, where a compromised authority may decrypt the ciphertext even in the absence of the requisite attribute keys. 

In addition to the above, there are a few other proposals \cite{yang2018improving, zhao2022Secure} in this area that, unfortunately, give rise to additional issues. For instance, a system developed by Yang et al. in \cite{yang2018improving} keeps the user's identity private from the attribute authority (\textit{AA}) if they are not in the same domain. Yet, this approach creates a new privacy issue that users might request \textit{AA}s within the domain to ask secret attribute keys from other \textit{AA}s on their behalf, implying that an \textit{AA} could potentially possess a complete set of a \textit{DU}'s secret keys. Zhao et al. presented a data sharing scheme \cite{zhao2022Secure} that adopts the access policy of linear secret sharing scheme (LSSS) and supports multi-authority CP-ABE scheme with policy-hiding to achieve the privacy-preserving functionality. However, this system is vulnerable to user key abuse due to its dependence on a single central authority for key generation.

\section{Preliminary}\label{sec:pre}
\newtheorem{assumption}{Assumption}
\newtheorem{myclaim}{\textbf{Claim}}
\newtheorem{definition}{Definition}

To initiate, we revisit certain foundational principles employed within our system. A summary of crucial notations utilized throughout the manuscript is provided in Table \ref{tab:Notation}.

\begin{table}[htb]
\caption{Notation Description}
\label{tab:Notation}
\begin{tabularx}{\columnwidth}{lX}
\hline
\textbf{Notation} & \textbf{Description}                                          \\ \hline

$\mathbb{G}_1, \mathbb{G}_2$ & Two additive cyclic groups of prime order $p$      \\
$\mathbb{G}_T$    & A multiplicative cyclic group of the same order $p$. \\
$\mathbb{Z}_p$    & A set of integers modulo $p$ \\
$\lambda$         & A security parameter that measures the input size of the system \\
$p$               & A prime number used for groups $\mathbb{G}_1, \mathbb{G}_2, \mathbb{G}_T$ and $\mathbb{Z}$\\
$i,j$             & Two indexes used to represent $i$th ($j$th) element or position in a sequence or set\\
$k$               & The parameter associated with the \textit{k-lin} assumption, representing the linear independence of group elements. \\

$PP_{ABE}$        & Public parameters for the use of Attribute-Based Encryption   \\
$PP_{VC}$         & Public parameters for the use of Vector Commitment            \\

$\alpha$          & A scalar used for generation of $PP_{ABE}$ and $PP_{VC}$      \\

$e$               & A set of secret elements used for \textbf{Trusted Setup}      \\
$e'$              & A set of secret elements used for \textbf{Authority Setup}    \\

$h$               & A hash of committed elements in \textbf{Trusted Setup}         \\
$h'$              & A hash of committed elements in \textbf{Authority Setup}       \\

$\pi$             & A proof of knowledge for an element                           \\
$\bm{rp}_s$       & A \textit{s-pair} of the element $s$ in group $\mathbb{G}_1$. The superscript $2$ of $\bm{rp}_s^2$ represent \textit{s-pair} elements in group $\mathbb{G}_2$ \\
$L$               & A list of \textit{s-pair} consisting of all the committed group elements \\

$PK$              & A public key owned by AA, which is used for ABE encryption    \\
$SK$              & A secret key owned by AA, which is used for ABE encryption    \\
$\bm{X}$          & A secret matrix element in $SK$                               \\
$\bm{\tau}$       & A secret vector element in $SK$                               \\
$\sigma$          & A secret value in $SK$                                        \\

$\bm{A},\bm{U}$   & The secret exponents used in $PP_{ABE}$                       \\
$K$               & A component of the attribute key generated by the attribute authority corresponds to an individual attribute.             \\
$sk$              & The consolidated secret key issued by an attribute authority. Given that an attribute authority can oversee multiple attributes, $sk$ might comprise several $K$ components, each representing a distinct attribute.     \\

$\bm{x}$          & A policy vector                                               \\
$\bm{v}$          & An attribute vector  \\

$\bm{C}$          & A Vector Commitment associated with a specific Data User, derived from its attribute vector and global identifier                                   \\
$m$               & A special message used in Vector Commitment                   \\
$op$              & An opening proof to reveal the Vector Commitment              \\
$o_i,o_{i,j}$     & The elements in $PP_{VC}$ where $i,j\in[n], i \neq j$         \\
$z$               & A secret exponent of group element $o$                        \\
$\bm{aux}$        & A collection of message $m$                                   \\


$\mathcal{U}$     & A set of attribute authorities in the universe                \\
$n$               & The number of attribute authorities includes $\textit{AA}_{\text{trust}}$ in the system             \\
$\mathcal{X}$     & A set of attributes in the universe                           \\
$l$               & The number of supported attributes                            \\
$\mathcal{S}$     & A set of attributes possessed by each attribute authority     \\
$\mathcal{R}$     & A set of attributes possessed by data user                    \\


$F$               & A file shared by data owner                                   \\ 
$AK$              & An AES key used to encrypt the file $F$                       \\
$CT_F$            & An encrypted sharing file $F$ encrypted by AES system         \\
$M$               & A metadata consisting related information of $CT_F$           \\
$kw$              & A keyword used in metadata to ease the data retrieval process \\
$CT_M$            & An encrypted metadata $M$ encrypted by ABE system             \\

$loc$             & A location address in IPFS for a file $F$                     \\
$addr$            & A blockchain address                                          \\
$GID$             & A Data User's Global identifier                               \\

$BPK$             & A public key registered in a blockchain                       \\
$BSK$             & A private key registered in a blockchain                      \\
IPFS              & The InterPlanetary File System                                \\\hline
\end{tabularx}
\end{table}

\subsection{Bilinear Mapping}

Consider $\mathcal{G}$ as an algorithm that accepts a security parameter $\lambda$ and constructs three multiplicative cyclic groups of prime order $p$: $\mathbb{G}_1 = \langle g_1 \rangle, \mathbb{G}_2 = \langle g_2 \rangle$, and $\mathbb{G}_T$. We introduce $\hat{e}$ as a bilinear map, with $\hat{e}: \mathbb{G}_1 \times \mathbb{G}_2 \rightarrow \mathbb{G}_T$. The bilinear map $\hat{e}$ has the following characteristics:

\begin{enumerate}
    \item \textbf{Bilinearity}: for all $a,b \in \mathbb{Z}$, $\hat{e}(g_1^a, g_2^b) = \hat{e}(g_1, g_2)^{ab}$.
    \item \textbf{Non-degeneray}: $\hat{e}(g_1, g_2) \neq 1$.
    \item \textbf{Computability}: for all $a,b \in \mathbb{Z}$, $\hat{e}(g_1^a, g_2^b)$ can be efficiently computed.
\end{enumerate}

\subsection{Auxiliary methods and definitions}

We make an assumption of possessing an algorithm, denoted as \textsc{COMMIT}, which takes strings of arbitrary length as input and produces outputs as determined by a random oracle. While this assumption aids our security analysis, in practical implementations, we could use the \textbf{BLAKE-2} hash function in place of \textsc{COMMIT}.


\begin{algorithm}
\caption{Commit of string $h$}\label{alg:commit}
\begin{algorithmic}[1]

\Require $h$ is a string

\Function{COMMIT}{$h$}
    \State $output \gets$ maps $h$ to random integer in $\mathbb{Z}_p^*$
    \State \Return $output$

\EndFunction
\end{algorithmic}
\end{algorithm}

For the inputs $h$ that can not be mapped directly to integers, especially in the case of group elements, we represent them using byte-strings. 

Additionally, we introduce several auxiliary methods to facilitate the verification procedure for certain special properties.

The following definitions and claims are first proposed in the work \cite{bowe2018multi}

\begin{definition} \label{def:sratio}
    Given a bilinear mapping $\hat{e}: \mathbb{G}_1 \setminus \{0\} \times \mathbb{G}_2 \setminus \{0\} \rightarrow \mathbb{G}_T$, elements $A,B \in \mathbb{G}_1 \backslash \{0\}$ and $C,D \in \mathbb{G}_2 \backslash \{0\}$. If $\hat{e}(A,D) = \hat{e}(B,C)$, we may use the term $\mathbf{SameRatio}((A,B),(C,D))$ to represent this relation. 
\end{definition}

\begin{algorithm}
\caption{Determin if two pairs $(A,B)$ and $(C,D)$ have certain relationship}\label{alg:SameRatio}
\begin{algorithmic}[1]

\Require $A,B \in \mathbb{G}_1 \backslash \{0\}$ and $C,D \in \mathbb{G}_2 \backslash \{0\}$

\Function{SameRatio}{$(A,B),(C,D)$}
    \If{$\hat{e}(A,D) = \hat{e}(B,C)$}
        \State \Return true
    \Else
        \State \Return false
    \EndIf
\EndFunction
\end{algorithmic}
\end{algorithm}

\begin{definition}  \label{def:spair}
Given a bilinear mapping $\hat{e}: \mathbb{G}_1 \setminus \{0\} \times \mathbb{G}_2 \setminus \{0\} \rightarrow \mathbb{G}_T$, $s\in \mathbb{Z}_p^*$ and cyclic group of order $p$, an \textit{s-pair} is a pair$(A,B)$ such that $A,B \in \mathbb{G}_1$, or $A,B \in \mathbb{G}_2$; and $s \cdot A = B$. For such an \textit{s-pair} $(A,B)$ in $\mathbb{G}_1$ or $\mathbb{G}_2$, we may represent it using the notation $\bm{rp}_s$ or $\bm{rp}_s^2$ respectively.
\end{definition}

\begin{myclaim}
\textbf{SameRatio} $((A,B),(C,D)) = $ true) if and only if there exists $s$ such that $(A,B)$ is an \textit{s-pair} in $\mathbb{G}_1$ and $(C,D)$ is an \textit{s-pair} in $\mathbb{G}_2$.
\end{myclaim}

\begin{proof}
Suppose there exist one element $s'$ that $s \cdot A = B$, $s' \cdot C = D$ and $s \neq s' \pmod{p}$. For some $a, c \in \mathbb{Z}_p^*$, we have $A = a \cdot g_1$ and $C = c \cdot g_2$. As defined in Definition \ref{def:sratio}, none of the elements $(a,c,s,s')$ is $0$.

Therefore,
\begin{gather*}
    \hat{e}(A,D) = \hat{e}(a \cdot g_1, s' \cdot C) = g_T^{acs'}\\
    \hat{e}(B,C) = \hat{e}(s \cdot A, c \cdot g_2) = g_T^{acs}
\end{gather*}

if and only if $s = s' \pmod{p}$, otherwise, $\mathbf{SameRatio}((p,q),(f,H))=$ false. As a result, no such $s'$ exists. 
\end{proof}

Finally, we can construct our special \textit{s-pair} as follows

\begin{definition} \label{def:sp-spair}
Given a bilinear mapping $\hat{e}: \mathbb{G}_1 \setminus \{0\} \times \mathbb{G}_2 \setminus \{0\} \rightarrow \mathbb{G}_T$ and a matrix $s \in \mathbb{Z}_p^{l \times k}$, a special \textit{s-pair} is a pair $(A, B)$ such that $A,B \in \mathbb{G}_1^{l \times k}$ or $A,B \in \mathbb{G}_2^{l \times k}$; and 
\begin{equation*}
    B[i,j] = A[i,j]^{s[i,j]}
\end{equation*} For such a special \textit{s-pair} $(A,B)$ in $\mathbb{G}_1$ or $\mathbb{G}_2$, we may also denote it as $\bm{rp}_{s}$. Given that a vector can be considered a matrix with a single column, we can also use the notation $\bm{rp}_{s}$ to represent an \textit{s-pair} when $s \in \mathbb{Z}_p^k$.

\end{definition}

\subsection{Assumptions}

Given a bilinear mapping $\hat{e}: \mathbb{G}_1 \setminus \{0\} \times \mathbb{G}_2 \setminus \{0\} \rightarrow \mathbb{G}_T$ with associated generators $\{g_1, g_2, g_T\}$ and group order $p$, our work builds upon a variety of standard assumptions, which are detailed below.

\begin{assumption} \label{asp:SXDH}
\textit{Symmetric External Diffie-Hellman (SXDH) assumption} \cite{ballard2005correlation}\\
It is hard to distinguish $\mathcal{D}_0 = (g_1, g_2, g_1^a, g_1^b, g_1^{ab})$ from $\mathcal{D}_1 = (g_1, g_2, g_1^a, g_1^b, g_1^c)$ where $a,b,c \xleftarrow[]{\$} \mathbb{Z}_p$. This also holds to the tuples $\mathcal{D}_0' = (g_1, g_2, g_2^a, g_2^b, g_2^{ab})$ and $\mathcal{D}_1' = (g_1, g_2, g_2^a, g_2^b, g_2^c)$ in different group.
\end{assumption}

\begin{assumption} \label{asp:klin}
\textit{K-Linear assumption} \cite{boneh2004short}\\
It is hard to distinguish $\mathcal{D}_0 = (g_1, g_2, g_1^{a_1}, g_1^{a_2}, \dots, g_1^{a_k}, g_1^{a_1b_1}, g_1^{a_2b_2},\dots,g_1^{a_kb_k}, g_1^{b_1 + b_2 + \dots + b_k})$ from $\mathcal{D}_1 = (g_1, g_2, g_1^{a_1}, g_1^{a_2}, \dots, g_1^{a_k}, g_1^{a_1b_1}, g_1^{a_2b_2},\dots,g_1^{a_kb_k}, g_1^{c})$ where $a_1,\dots,a_k, b_1,\dots,b_k, c \xleftarrow[]{\$} \mathbb{Z}_p$. This also holds in the group $\mathbb{G}_2$. 
\end{assumption}

For $a_1, a_2, ... , a_k \xleftarrow{\$} \mathbb{Z}_p^*$, we have the construction of

\begin{equation} \label{eq:A}
    \bm{A} = \begin{pmatrix}
    a_1 & 0 & & 0 \\
    0   & a_2 & & 0\\
    \vdots & \ddots &  & \\
    0 & 0 & \cdots & a_k\\
    1 & 1 & \cdots & 1 \\
    \end{pmatrix} \in \mathbb{Z}^{(k+1)\times k}_p
\end{equation} and 
\begin{equation} \label{eq:a}
    \bm{a}^{\perp} = \begin{pmatrix}
    a_1^{-1}\\
    a_2^{-1}\\
    \vdots\\
    a_k^{-1}\\
    -1\\
    \end{pmatrix} \in \mathbb{Z}^{(k+1)}_p
\end{equation}, with $\bm{A^{\intercal} a}^{\perp} = 0$. 

\begin{assumption} \label{asp:sp-klin}
\textit{Special k-Linear assumption} \cite{michalevsky2018decentralized}\\
Given a randomly generated matrix $\bm{A} \in \mathbb{Z}_p^{(k+1) \times k}$ and a vector $\bm{s} \in \mathbb{Z}_p^{(k+1)}$, the tuples $\mathcal{D}_0 = (g_1, g_2, g_1^{\bm{A}}, g_1^{\bm{As}})$ and $\mathcal{D}_1 = (g_1, g_2, g_1^{\bm{A}}, g_1^{\bm{s}})$ are computationally indistinguishable by any polynomial-time $\mathcal{A}$. The structure of matrix $\bm{A}$ is described in Eq.\ref{eq:A} and the vector $\bm{a}$ is derived from $\bm{A}$ as detailed in Eq.\ref{eq:a}.
\end{assumption}

\begin{assumption} \label{asp:s-cdh}
\textit{Square Computational Diffie-Hellman assumption} \cite{burmester1998equitable}\\
Given $(g, g^a)$ for a random number $a$ in a cyclic group $\mathbb{G}$ of order $p$, a $PPT$ algorithm $\mathcal{A}$ outputs $g^{a^2}$ with non-negligible probability.
    
\end{assumption}

\begin{assumption} \label{asp:KCA}
\textit{Knowledge of Coefficient assumption} \cite{bowe2017scalable} \\
Given a string of arbitrary length $h$, and a uniformly chosen $C \in \mathbb{G}_2$ (independent of $h$), an efficient algorithm $\mathcal{A}$ exists that can randomly generate $B \in \mathbb{G}_1$ and $D \in \mathbb{G}_2$. Meanwhile, for the same inputs $(C,h)$, there is an efficient deterministic algorithm $\mathcal{X}$ cable of extracting a scalar $b$. The probability that:

\begin{enumerate}
    \item $\mathcal{A}$ `succeeds', meaning it satisfies the condition: SameRatio($(g_1, B),(C,D)$)
    \item $\mathcal{X}$ `fails', meaning $B \neq g_1^{b}$
\end{enumerate}

is negligible.

\end{assumption}

\subsection{Proof of Knowledge}
We adopt the well-established Schnorr Identification Protocol \cite{schnorr1990efficient}, utilizing it as our Non-interactive Zero-knowledge Proof (NIZK). Provided with an \textit{s-pair} $rp_s = (A, B = s \cdot A)$ and a string $h$, we establish \textsc{NIZK} (Algorithm \ref{alg:NIZK}). This can serve as a proof that the originator of the string $h$ is aware of $s$ in the \textit{s-pair} $rp_s$. 

\begin{algorithm}
\caption{Construct a proof of knowledge of $s$}\label{alg:NIZK}
\begin{algorithmic}[1]

\Require $rp_s$ is an \textit{s-pair}
\Require $h$ is a string

\Function{NIZK}{$rp_s = (A,B), h$}
    \State $\alpha \xleftarrow[]{\$} \mathbb{Z}_p^*$
    \State $R \gets \alpha \cdot A$
    \State $c \gets$ \textbf{COMMIT}$(R || h) \in \mathbb{Z}_p^*$
    \State $u \gets \alpha + c \cdot s$
    
    \State \textbf{return} $\pi = (R, u)$

\EndFunction

\end{algorithmic}
\end{algorithm}

Furthermore, we define \textsc{VerifyNIZK} (Algorithm \ref{alg:vrfNIZK}), which checks the validity of the provided proof $\pi$.

\begin{algorithm}
\caption{Verify a proof of knowledge of $s$}\label{alg:vrfNIZK}
\begin{algorithmic}[1]

\Require $rp_s$ is an \textit{s-pair}
\Require $h$ is a string

\Function{VerifyNIZK}{$rp_s = (A,B), \pi=(R,u), h$}
    \State $c \gets$ \textbf{COMMIT}$(R || h) \in \mathbb{Z}_p^*$
    \If{$u \cdot A == R + c \cdot B$}
        \State \textbf{return} true
    \Else
        \State \textbf{return} false
    \EndIf

\EndFunction
\end{algorithmic}
\end{algorithm}

\subsection{Vector Commitment} 
We ensure our attribute-hiding property through the utilization of a Vector Commitment scheme, as described in \cite{catalano2013vector}. The summarized scheme is as follows:


\begin{definition}\label{def:VC}
This Vector Commitment system commits to an ordered sequence of attribute elements $\bm{v} = (v_1, v_2, \cdots, v_{l+1})$ as commitment $\bm{C}$, then opens it in a certain position of $\bm{v}$ to a corresponding attribute authority (\textit{AA}), and finally proves that only authorized value existed in the previously supplied commitment $\bm{C}$. The system normally consists of 4 algorithms:
\end{definition}

\begin{itemize}
    \item \textbf{Key Generation} $(1^\lambda, n) \rightarrow PP_{VC}$: This is a Decentralized Key Generation (DKG) algorithm. It takes as input the security parameter $\lambda$ and the number of attribute authorities, $n$, in the system, and outputs global public parameters $PP_{VC} = \{g_1, g_2, \{o_i\}, \{o_{i,j}\} \}$ where $i,j \in [n], i\neq j$. The element $o_i$ is generated and published by $\textit{AA}_i$. Following that, the elements $\{o_{i,j}\}$ can be issued by each $\textit{AA}_i$ based on the shared $\{o_i\}$.
    
    \item \textbf{Commitment} $(\bm{aux} = \{m_i\}_{i\in [n]}) \rightarrow \bm{C}$: This algorithm is run by a Data User (\textit{DU}). It takes as input the message $m_i$ generated based on the authorized attributes from $\textit{AA}_i, i\in [n]$, and outputs the commitment $\bm{C}$. 
    \item \textbf{Open} $(m_i, i, \bm{aux}) \rightarrow op_i$: This algorithm is also run by a \textit{DU}. It takes as input the auxiliary information $\bm{aux}$ and index $i$, and outputs the opening proof $op_i$.
    \item \textbf{Verify} $(\bm{C}, m_i, i, op_i) \rightarrow (1 ~or~ 0)$: The Verify algorithm is run by \textit{AA}. It takes as inputs the commitment $\bm{C}$, message $m_i$, index $i$, and opening proof $op_i$, and outputs the result of the verification. It outputs $1$ when it accepts the proof. 
\end{itemize}

\subsection{Decentralized Inner-Product Predicate Encryption}


\begin{definition}\label{def:IPPE}
 A Multi-authority Attribute-based Encryption with policy-hiding scheme \cite{michalevsky2018decentralized} consists of a tuple of Probabilistic Polynomial-Time (PPT) algorithms, such that:
 \end{definition}
 
\begin{itemize}
    \item \textbf{Setup} $(1^\lambda) \rightarrow PP$: The global setup algorithm is a decentralized generation algorithm that takes as input the security parameter $\lambda$ and then outputs the public parameters $PP$.
    
    \item \textbf{Authority Setup} $(PP, i) \rightarrow (PK_i,SK_i)$: The authority setup algorithm is run by each $\textit{AA}_i$. It takes as input public parameter $PP$ and authority index $i$, and outputs a pair of authority keys $(PK_i, SK_i)$.

    \item \textbf{Key Generation} $(PP, i, SK_i, \{PK\}, GID, \bm{v}) \rightarrow sk_{i, GID, \bm{v}}$: The key generation algorithm is run by each $\textit{AA}_i$. It takes as input the global public parameters $PP$, the authority index $i$, its secret key $SK_i$, all the public keys $\{PK_i\}_{i \in [n]}$, and \textit{DU}'s global identifier $GID$ and the attribute vector $\bm{v}$, and outputs the secret keys $sk_{i,GID,\bm{v}}$.

    \item \textbf{Encryption} $(PP, \{PK\}, \bm{x}, F) \rightarrow CT_F$: The data encryption algorithm is run by a Data Owner (\textit{DO}). It takes as inputs the global parameters $PP$, the public keys of all the authorities $\{PK_i\}_{i \in [n]}$, the ciphertext policy vector $\bm{x}$ and a file $F$, and outputs a ciphertext $CT_F$.
    
    
    \item \textbf{Decryption} $(\{sk_{i, GID, \bm{v}} \}_{i\in n}, CT_F) \rightarrow F$: The decryption algorithm is run by the authorized \textit{DU}. It takes as inputs the collection of secret keys $\{sk_{i, GID, \bm{v}}\}$ from $\textit{AA}_i$ and the ciphertext $CT_F$, and outputs the message $F$ if the access policy has been satisfied.
    
\end{itemize}

It's important to highlight that the algorithms delineated above do not hide the attributes vector $\bm{v}$ from the authorities. Rather, we incorporate the vector commitment $\bm{C}$ as an input for the \textbf{Key Generation} process. The detail of this procedure will be provided in the context of Section \ref{sec:design}.

\section{System Overview}\label{sec:overview}
\subsection{System Architecture}
The system architecture is shown in Figure \ref{fig:p1}, 
\begin{figure*}[h!]\centering
\includegraphics[width=\textwidth]{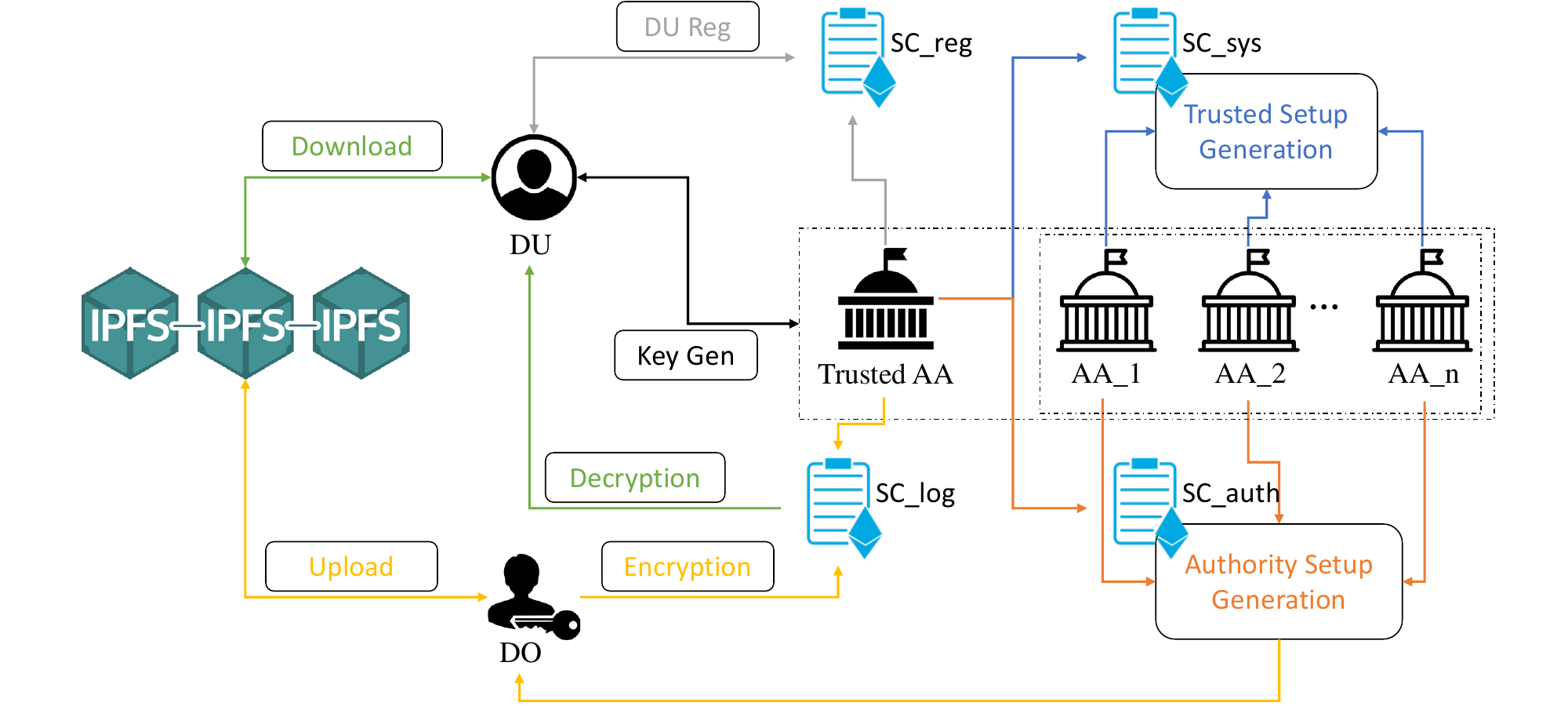}
\caption{The system consists of 6 processes, each of which is represented by a different color: Blue for the process \textbf{Trusted Setup}, orange for the process \textbf{Authority Setup}, gray for the process \textbf{Data User Registration}, black for the process \textbf{Key Generation}, yellow for the process \textbf{Encryption and Upload}, and green for the process \textbf{Download and Decryption}. And these connectors, Single or Double-Arrow, indicate the interactions between service users and two blockchain networks, Ethereum and IPFS. Note that these 4 contracts deployed on Ethereum are used for data governance, while IPFS is used for data storage. For a detailed description of the system flow between smart contracts, IPFS, and various entities, please check Section \ref{sec:overview} for the Interaction Overview and Section \ref{sec:design} for the System Design.} \label{fig:p1}
\end{figure*}
which comprises the following logical entities:


\textbf{Data Owner} (\textit{DO}): \textit{DO} is an entity (individual or organization) that owns a certain file $F$. For secure storage and sharing, \textit{DO} encrypts $F$ using the AES key $AK$ and uploads the encrypted file $CT_F$ to the IPFS network, records the returned file location $loc$, and embeds $AK$ and $loc$ into the metadata $M$ which is subsequently encrypted using the ABE system and published $CT_M$ in the Ethereum network. 

\textbf{Data User} (\textit{DU}): \textit{DU} is a data client for \textit{DO}. It asks the attribute authority \textit{AA} for permission to get the necessary attribute secret keys $\{sk\}$, which are then used to decrypt the associated $CT_M$ stored on the Ethereum network. After getting the key $AK$ and the location $loc$ from $M$, \textit{DU} can download the encrypted file $CT_F$ from the IPFS network and recover the original file $F$.

\textbf{Attribute Authority} (\textit{AA}): \textit{AA} is an entity (individual or organization) that contributes to the generation of the public parameters of the ABE system $PP_{ABE}$ and the vector commitment scheme $PP_{VC}$, publishes the public key $PK$ for the \textit{DO} to encrypt metadata $M$, owns a set of attributes and issues secret key $sk$ for the owned attributes upon the request of the \textit{DU}.

\textbf{Trusted Attributed Authority} ($\textit{AA}_{\text{trust}}$): $\textit{AA}_{\text{trust}}$ is a trusted attribute authority that mainly generates a secret key $sk$ for \textit{DU} and deploys system contracts for setup and registration. $\textit{AA}_{\text{trust}}$, unlike normal \textit{AA}, owns no attributes but is in charge of a specific position in the attribute vector $\bm{v}$. 

\textbf{Service User} (\textit(SU)): In the system, \textit(SU) is a general entity comprising \textit{DO}, \textit{DU} and \textit{AA}. 

\textbf{Participant} ($P$): $P$ is a special entity that represents each \textit{AA} during the process of \textbf{Trusted Setup}. The index $i$ of $P_i$ denotes the chronological order of each piece of public parameter generated and shared by \textit{AA}. 

\textbf{Trust Setup Contract} ($SC_{sys}$): The contract $SC_{sys}$ is deployed to the Ethereum network by the $\textit{AA}_{\text{trust}}$ and can only be invoked by an authorized \textit{AA} within the time window specified. It is responsible for generating the global public parameters $PP_{ABE}$.

\textbf{Authority Setup Contract} $(SC_{auth})$: Contract $SC_{auth}$ is deployed to the Ethereum network by the $\textit{AA}_{\text{trust}}$. It can only be invoked by the authorized \textit{AA} within the specified time window. It is used to generate the global public parameters $PP_{VC}$ and to keep track of the valid information about \textit{AA}'s address $addr$, public key $PK$, and supported attributes $\mathcal{S}$. 

\textbf{User Registration Contract} $(SC_{reg})$: Contract $SC_{reg}$ is deployed to the Ethereum network by the $\textit{AA}_{\text{trust}}$ and can be invoked by all the potential \textit{DU}s. To register the $addr$ in the system, \textit{DU} needs to make sufficient payment to the $SC_{reg}$ and then get back the $GID$ which can later be used to request secret key $sk$ from \textit{AA}. 

\textbf{Utility Contract} ($SC_{util}$): Contract $SC_{util}$ is deployed to the Ethereum network by the $\textit{AA}_{\text{trust}}$ and can only be invoked by other  contracts deployed by $\textit{AA}_{\text{trust}}$. It is mainly used to verify group elements published by \textit{AA}.

\textbf{Log Contract} ($SC_{log}$): Contract $SC_{log}$ is deployed to the Ethereum network by the $\textit{AA}_{\text{trust}}$. When it receives a new transaction from \textit{DO}, it records the encrypted data $CT_M$ of metadata $M$ and triggers the event to the subscribers.  

\textbf{Blockchain}: Each user (\textit{DO, DU, AA} and $\textit{AA}_{\text{trust}}$) possesses a pair of keys $(BPK, BSK)$ and a corresponding wallet address $addr$ on the blockchain. Our system employs two blockchains: IPFS for data storage and Ethereum for data governance.

\subsection{Interactions overview} \label{overview}

In this section, we describe the overview of our proposed system to show how smart contracts, IPFS, vector commitment, and MA-CP-ABE with policy-hiding are composed together to build a secure, privacy-preserving, and blockchain-enabled data governance system. When it ought to be clear from the context, we omit most indices, like $i$ and $j$ of elements, and superscript in $\bm{rp}_s^2$ for readability.

\renewcommand{\labelenumii}{\arabic{enumii}}

\begin{enumerate}
    \item \textbf{Trusted Setup} 
    \begin{enumerate}
        \item First, a community of normal attribute authorities (\textit{AA}s) with size $n-1$ and a special trusted attribute authority ($\textit{AA}_{\text{trust}}$) must be determined. $\textit{AA}_{\text{trust}}$ selects the security parameter $\lambda$ and two generators $g_1, g_2$ for the bilinear mapping, and defines following algorithms: \textsc{COMMIT}, \textsc{NIZK}, \textsc{VerifyNIZK} and \textsc{powerMulti}.
        
        \item $\textit{AA}_{\text{trust}}$ deploys one system contract $SC_{sys}$ and one utility contract $SC_{util}$. 
        
        \item Each \textit{AA} randomly samples a set of secret elements $e$: two matrixes $\bm{A} \in \mathbb{Z}_p^{(k+1) \times k} $ and $ \bm{U}\in \mathbb{Z}_p^{(k+1) \times (k+1)}$, two scalars $\alpha_{\bm{A}} $ and $ \alpha_{\bm{U}}$ in $\mathbb{Z}_p^*$, and two scaled matriex $\alpha_{\bm{A}} \cdot \bm{A} $ and $ \alpha_{\bm{U}} \cdot \bm{U}$, and publishes a corresponding set of \textit{s-pair} $\{\bm{rp}_{\bm{A}}, \bm{rp}_{\bm{U}}, \bm{rp}_{\alpha_{\bm{A}}}, \bm{rp}_{\alpha_{\bm{U}}}, \bm{rp}_{\alpha \cdot \bm{A}}, \bm{rp}_{\alpha \cdot \bm{U}}\}$ as defined in Def.\ref{def:spair} and Def.\ref{def:sp-spair} to the contract $SC_{sys}$.
        
        \item After that, \textit{AA} computes and publishes the commitments $h := \text{\textsc{COMMIT}}(\{h_{s}\} ||)$ as showed in Alg. \ref{alg:commit} to $SC_{sys}$, where $h_{s} := \text{\textsc{COMMIT}}(\bm{rp}_s), s \in e$.

        \item Every \textit{AA} then needs to prove the knowledge of each element $s \in e$ by outputting the proofs $\{\pi_s\}$ using algorithm \textsc{NIZK} (Alg. \ref{alg:NIZK}) as the argument of function \textsc{Prove} of contract $SC_{sys}$, which verifies them using algorithm \textsc{VerifyNIZK} (Alg. \mbox{\ref{alg:vrfNIZK}}).

        \item In the \textbf{Round 1}, we define one attribute authority \textit{AA} as participant $P_1$ who firstly publishes group elements in $\mathbb{G}_1$: $(V_1 :=g_1^{\bm{A}_1}, \theta_{V_1} := g_1^{\alpha_{\bm{A}_1}}, V_1' := g_1^{\alpha_{\bm{A}_1} \cdot \bm{A}_1} )$, based on the previously verified set of elements $e$. 
        
        \item Participant $P_{i = 2, \dots, n}$ computes $(V_i, \theta_{V_i},V_i')$ based on previous $( V_{i-1}, \theta_{V_{i-1}}, V_{i-1}' )$ using algorithm \textsc{powerMult} (Alg. \ref{alg:pMu}) and publishes these as the arguments of the function \textsc{Compute} of contract $SC_{sys}$ to check validity. 
        
        \item We define the last valid $V$ received by contract $SC_{sys}$ as one piece of the public parameter $g_1^{\bm{A}}$.

        \item In the \textbf{Round 2}, the first \textit{AA}, also known as participant $P_1$, publishes group elements in $\mathbb{G}_1$: $( W_1 := ((g_1^{\bm{A}})^{\intercal})^{\bm{U}_1}, \theta_{W_1} := g_1^{\alpha_{\bm{U}_1}}, W_1' := ((g_1^{\bm{A}})^{\intercal})^{\alpha_{\bm{U}_1} \bm{U}_1} )$, also based on the previously verified set of elements $e$.
        
        \item Participant $P_i$, where $i = 2, \dots, n$ computes its $( W_i, \theta_{W_i}',W_i' )$ based on previous $( W_{i-1}, \theta_{W_{i-1}}',W_{i-1}' )$ using algorithm \textsc{powerMult} and publishes these as the arguments of the function \textsc{Generate}. 
        
        \item We also define the last valid $W$ received by contract $SC_{sys}$ as last piece of the public parameter $g_1^{\bm{U}^{\intercal} \bm{A}}$. Therefore, we have $PP_{ABE} := \{g_1, g_2, g_1^{\bm{A}}, g_1^{\bm{U}^{\intercal} \bm{A}}\}$.
        
    \end{enumerate}
    \item \textbf{Authority Setup}
    \begin{enumerate}
        \item $\textit{AA}_{\text{trust}}$ deploys contract $SC_{auth}$ for authority setup.
        
        \item Each \textit{AA} randomly samples another set of secret element $e'$: a matrix $\bm{X} \in \mathbb{Z}_p^{(k+1) \times (k+1)}$, a vector $\bm{\tau} \in \mathbb{Z}_p^{k+1}$, two numbers $\sigma, z \in \mathbb{Z}_p^*$, a scalar $\alpha_{z}$ and a scaled number $\alpha_z \cdot z$. Using that, \textit{AA} takes $SK := \{ \bm{X}, \bm{\tau}, \sigma\}$ as secret keys, and publishes a corresponding set of \textit{s-pair} $\{ \bm{rp}_{\bm{X}}, \bm{rp}_{\bm{\tau}}, \bm{rp}_{\sigma}, \bm{rp}_{z}, \bm{rp}_{\alpha_z \cdot z}\}$ to the contract $SC_{auth}$.

        \item After that, \textit{AA} computes and publishes the commitment $h' := \text{\textsc{COMMIT}}(\{h_s\} ||)$ to $SC_{auth}$, where $h_s := \text{\textsc{COMMIT}}(\bm{rp}_s), s \in e'$. 

        
        \item Every \textit{AA} then needs to prove the knowledge of elements $s \in e'$ by generating the proofs $\{\pi_s\}$ using algorithm \textbf{NIZK} as the argument of function \textsc{Prove} of contract $SC_{auth}$ for validity check. 
        
        
        
        \item We define each \textit{AA} with index $i \in [n-1]$ based on the receiving order of the complete set of valid $\{\pi_s\}_{s \in e'}$ and set attribute authority $\textit{AA}_{\text{trust}}$ with index $n$.
        
        \item Therefore, we have the verified sets of elements $PK_i := \{g_1^{\bm{X}_i^{\intercal} \cdot \bm{A}}, \hat{e}(g_1^{\bm{\tau}_i^{\intercal}\bm{A}},g_2), g_2^{\sigma}\}$ and $\{o_i := g_1^{z_i}, g_1^{\alpha_{z_i}}, g_1^{\alpha_{z_i} \cdot z_i} \}$ for each $AA_i$.
        

        \item In the last stage, for each $i,j \in [n], j \neq i$, $AA_i$ needs to compute a set of group elements in $\mathbb{G}_1: (O_i:=\{(o_{j})^{z_i}\}, \theta_{O_i} :=\{(g_1^{\alpha_{z_j}})^{\alpha_{z_i}}\}, O'_i:= \{(g_1^{ \alpha_{z_j}\cdot z_j})^{\alpha_{z_i} \cdot z_i}\}$. Then $AA_i$ publishes these elements, with the number of supported attributes $l_i$ as the argument of the function \textsc{Setup}.

        \item The contract $SC_{auth}$ checks the validity of these elements published by $AA_i, i \in [n]$ and then registers its address $addr_i$ with the elements $(l_i, PK_i)$. 

        \item In the end, we have $PP_{VC} := \{g_1, g_2, \{o_i\}_{i \in [n]}, \{o_{i,j}\}_{i, j \in [n], i\neq j}\}$ for vector commitment scheme.
        
    \end{enumerate}
    \item \textbf{Data User Registration}

    \begin{enumerate}
        \item $\textit{AA}_{\text{trust}}$ deploys contract $SC_{reg}$ for service registration.
        \item Data User (\textit{DU}) makes a direct registration payment to the contract $SC_{reg}$ to get the global identifier $GID$ which is the hash value of DU's address $addr$.
        \item Afterwards, \textit{DU} can setup a secure channel with each $AA_i, i \in [n]$ that possesses the needed attributes and can verify \textit{DU}'s identity.

        \item $AA_i$ verifies \textit{DU}'s identity and sends back the set of acknowledged attributes $\mathcal{R}_{i,GID}$ through the secure channel.
        \item Upon receiving all the $\mathcal{R}_{i,GID}$ from $AA_i$, \textit{DU} defines a set of `N/A' attributes $\mathcal{R}_{j,GID}$ for those $AA_j, i \neq j$ can not issue the attribute set and finally gets a complete set of attributes $\mathcal{R}_{GID}$ by combing $\mathcal{R}_{i, GID}$ and $\mathcal{R}_{j, GID}$ together.

    \end{enumerate}
    
    \item \textbf{Key Gen}

    \begin{enumerate}
        \item \textit{DU} generates an attribute vector $\bm{v}_{GID}$ from set of attributes $\mathcal{R}_{GID}$, creates a vector commitment $\bm{C}$ for $\bm{v}_{GID}$ and sends it with opening proof $op_i$ to each $AA_i, i \in [n]$ through separate secure channels.

        \item $AA_i, i \in [n]$ firstly checks the validity of its responsible part in the commitment $\bm{C}$ using $op_i$, then issues \textit{DU}'s requested attribute secret key $sk_{i, GID, \bm{C}}$, and finally sends it back to \textit{DU} through the channel.

        \item Upon receiving responses from each $AA_i, i\in[n]$, \textit{DU} gets a complete set of secret keys $\{sk_{i,GID,\bm{C}}\}_{i \in [n]}$. 
    \end{enumerate}

    \item \textbf{Encryption and Upload} 
    \begin{enumerate}
        \item $\textit{AA}_{\text{trust}}$ deploys last system contract $SC_{log}$ to record encrypted related information of file $F$
        \item Data Owner (\textit{DO}) randomly samples an AES key $AK$, encrypts $F$ to obtain the ciphertext $CT_F$, and uploads it to the IPFS network. 
        \item After successfully receiving the $CT_F$ from \textit{DO}, IPFS network returns a special hash value $loc$ as a file location on the IPFS network.
        
        \item Then, \textit{DO} constructs a metadata $M := (K, loc)$, specifies a policy vector $\bm{x}$ based on selected attributes from each $AA_i$, uses published $\{PK_i\}$ to encrypt the metadata $M$ and publishes this encrypted information $CT_M$ to contract $SC_{log}$.

    \end{enumerate}
    
    \item \textbf{Download and Decryption}
    \begin{enumerate}
        \item \textit{DU} reads every new coming $CT_M$ from the contract $SC_{log}$ and checks if its owned secret keys $\{sk_{i, GID, \mathbf{C}}\}_{i \in [n]}$ satisfies the access policy $\bm{x}$ to recover the metadata $M$. 
        
        \item \textit{DU} retrieves the file location $loc$ and AES key $AK$ from the metadata $M$ and requests the ciphertext $CT_F$ from the IPFS network with the file location $loc$.
        \item \textit{DU} uses the AES key $AK$ to recover the original file $F$. 
    \end{enumerate}

\end{enumerate}
\section{System Design}\label{sec:design}
\newcommand{\INDSTATE}[1][1]{\STATE\hspace{#1\algorithmicindent}}

In this section, we provide more details on \textbf{Trusted Setup}, \textbf{Authority Setup}, \textbf{Data User Registration}, \textbf{Key Generation}, \textbf{Encryption and Upload}, and \textbf{Download and Decryption} processes.

\subsection{Trusted Setup} \label{sub_sec:trusted_setup}
The process of \textbf{Trusted Setup} consists of 4 stages: \emph{Initiate}, \emph{Commit and Reveal}, \emph{Verify}, and \emph{Generate}, and finally outputs the global public parameter $PP_{ABE}$ for the ABE system. 

In the initial three stages, each attribute authority \textit{AA} sends its transactions independently to the contract $SC_{sys}$. In contrast, during the final \emph{Generate} stage, each participant $P$ (where we use the placeholder notation $P$ in \emph{Commit and Reveal} to represent each \textit{AA}) must send transactions to $SC_{sys}$ in a sequential manner. This sequentiality is necessary because each incoming transaction is generated based on the preceding $P$'s transaction received by contract $SC_{sys}$.

\subsubsection{Initiate} 

At the start, a fixed-numbered community of size $n$ will be determined, which will include all of the normal attribute authorities \textit{AA} and one special trusted authority $\textit{AA}_{\text{trust}}$. $\textit{AA}_{\text{trust}}$ represents this community to set the global security parameter to be $\lambda$ and the generators of the $\mathbb{G}_1, \mathbb{G}_2$ with prime order $p$ to be $g_1, g_2$ respectively. Therefore, the bilinear map can be $\hat{e}: \mathbb{G}_1 \setminus \{0\} \times \mathbb{G}_2 \setminus \{0\} \rightarrow \mathbb{G}_T$. 

$\textit{AA}_{\text{trust}}$ also deploys two distinct system contracts: contract $SC_{sys}$ for trusted setup and contract $SC_{util}$ for resolving the problem of allowing complex cryptographic computations to be used in the system. $\textit{AA}_{\text{trust}}$ also specifies the deadlines $(ddl_1, ddl_2, ddl_3)$ for the process \textbf{Trusted Setup} and sets an authorized list $AAlist$ to restrict $SC_{sys}$ access. For a simple description of the system, we assume that each attribute authority \textit{AA} submits the required transactions within the deadlines. The pseudo codes of these two contracts are provided in Appendix \ref{sc-sys} and \ref{sc-util}, respectively.

To realize the generation of $PP_{ABE}$, this process highly depends on the interaction between each attribute authority \textit{AA} and contract $SC_{sys}$, which has five main functions, \textsc{Commit}, \textsc{Reveal}, \textsc{Prove}, \textsc{Compute} and \textsc{Generate} with the help from contract $SC_{utils}$. Generally, these functions can only be invoked by a blockchain address owned by \textit{AA}, which is included in the authorized list $AAlist$, and executed before the pre-defined deadline $ddl_{1,2,3}$. Some variables used here are listed below:

\begin{enumerate}
    \item $msg.sender$ (Global Variable): The sender of the message or transaction.
    \item $block.timestamp$ (Global Variable): The current block timestamp as seconds.
    
    \item $AAlist$ (State Variable): A list of pre-verified blockchain addresses belonging to attribute authorities set by a trusted authority. 
    \item $ddl_1$ (State Variable): A deadline that only allows transaction calls within the pre-set time window for the stage \emph{Commit and Reveal}.
    \item $ddl_2$ (State Variable): A deadline that only allows transaction calls within the pre-set time window for the stage \emph{Verify}.
    \item $ddl_3$ (State Variable): A deadline that only allows transaction calls within the pre-set time window for the stage \textit{Compute and Generate}.

\end{enumerate}

\subsubsection{Commit and Reveal} \label{C&R}
Every \textit{AA} randomly picks a set of elements $e$: a matrix $\bm{A} \xleftarrow{\$} \{ \text{Diagonal matrices in } \mathbb{Z}_p^{k \times k} \} \cup \begin{pmatrix} \mathbf{1} \end{pmatrix}$, a matrix $\bm{U} \xleftarrow{\$} \mathbb{Z}_p^{(k+1) \times (k+1)}$ \footnote{The value of $k$ in this context does not derive from the security parameter $\lambda$. Rather, it is a reference to the \textit{k-Lin} assumption \mbox{\ref{asp:klin}} upon which our construction is predicated.}, their corresponding scalar values, $\alpha_{\bm{A}}$ and $\alpha_{\bm{U}}$, and scaled matrixes $\alpha_{\bm{A}} \cdot \bm{A}, \alpha_{\bm{U}} \cdot \bm{U}$. 

Then, \textit{AA} has \begin{equation}
    e = \{\bm{A}, \bm{U}, \alpha_{\bm{A}}, \alpha_{\bm{U}}, \alpha_{\bm{A}} \cdot \bm{A}, \alpha_{\bm{U}} \cdot \bm{U} \}
\end{equation} and then generates a set of \textit{s-pair}.

For such element $s \in e$, we refer to the \textit{s-pair} in $\mathbb{G}_1$ by $\bm{rp}_s$ and in $\mathbb{G}_2$ by $\bm{rp}_s^2$ as Definition \ref{def:spair}. These \textit{s-pair} in both $\mathbb{G}_1$ and $\mathbb{G}_2$ are listed as follows:
\begin{itemize} \label{spairs}
    \item For matrix $\bm{A}$: $(\bm{rp_{\bm{A}}}, \bm{rp_{\bm{A}}}^2) = (g, g^{\bm{A}})$
    \item For matrix $\bm{U}$: $(\bm{rp_{\bm{U}}},\bm{rp_{\bm{U}}}^2) = (g^{\bm{A}^{\intercal}}, g^{\bm{A}^{\intercal}\bm{U}})$
    
    \item For scalar $\alpha_{\bm{A}}$: $(\bm{rp}_{\alpha_{\bm{A}}},\bm{rp}_{\alpha_{\bm{A}}}^2) = (g, g^{\alpha_{\bm{A}}})$
    \item For scalar $\alpha_{\bm{U}}$: $(\bm{rp}_{\alpha_{\bm{U}}},\bm{rp}_{\alpha_{\bm{U}}}^2) = (g, g^{\alpha_{\bm{U}}})$
    
    \item For scaled matrix $\alpha_{\bm{A}} \bm{A}$: $(\bm{rp}_{\alpha_{\bm{A}} \bm{A}},\bm{rp}_{\alpha_{\bm{A}} \bm{A}}^2) = (g, g^{\alpha_{\bm{A}} \cdot \bm{A}})$
    \item For scaled matrix $\alpha_{\bm{U}}\bm{U}$: $(\bm{rp}_{\alpha_{\bm{U}} \bm{U}},\bm{rp}_{\alpha_{\bm{U}} \bm{U}}^2) = (g^{\bm{A}^{\intercal}}, g^{\alpha_{\bm{U}} \cdot \bm{A}^{\intercal}\bm{U}})$
\end{itemize}

Other than these \textit{s-pair} listed above, \textit{AA} also commits each element $s \in e$ using Algorithm \ref{alg:commit}. For each $s \in e$:
\begin{equation*}
    h_{s} := \textsc{COMMIT}(\bm{rp}_s || \bm{rp}_s^2)
\end{equation*}

Subsequently, the overall commitment is:
\begin{equation}
    h := \textsc{COMMIT}(h_{\bm{A}} || h_{\bm{U}} || h_{\bm{\alpha_{\bm{A}}}} ||
    h_{\alpha_{\bm{U}}} || h_{\alpha_{\bm{A}} \bm{A}} || h_{\alpha_{\bm{U}} \bm{U}})
\end{equation}

After that, \textit{AA} publishes the commitment $h$ to the contract $SC_{sys}$ through blockchain transaction. The state variable $h\_collector$ of $SC_{sys}$ (Algorithm \ref{SS1}) will store the value $h$ with the key as $msg.sender$, also known as \textit{AA}'s blockchain address. Apart from $h\_collector$, we define few state variables used in contract $SC_{sys}$ as follows:

\begin{enumerate}

    \item $h\_collector$ (State Variable): A mapping collection from the blockchain address belonged to one attribute authority to its commitment $h$
    \item $unverified\_elements$ (State Variable): A mapping collection from the blockchain address belonged to one attribute authority to its unverified list of \textit{s-pair}s $L = \{(\bm{rp}_s, \bm{rp}_s^2) | s \in e\}$ in both $\mathbb{G}_1$ and $\mathbb{G}_2$
    \item $verified\_elements$ (State Variable): A mapping collection from the blockchain address belonged to one attribute authority to its verified list of \textit{s-pair}s $L = \{(\bm{rp}_s, \bm{rp}_s^2) | s \in e\}$ in both $\mathbb{G}_1$ and $\mathbb{G}_2$
    
\end{enumerate}  

After $h$ has been received by contract $SC_{sys}$, the sender needs to reveal committed element $s \in e$ by passing a list of \textit{s-pair} in both $\mathbb{G}_1$ and $\mathbb{G}_2$
\begin{equation}
    L_{rp} = \{(\bm{rp}_s, \bm{rp}_s^2) | s \in e\}
\end{equation} as argument of the function \textsc{Reveal} (Algorithm \ref{SS2}) before deadline $ddl_1$, which checks the existence of the $h$ published by $msg.sender$, and then verify that indeed $h = \textsc{COMMIT}(\{h_{s}\} ||)$ \footnote{For brevity, we will use this shorthand notation to represent the above concatenation where $s$ takes on all value in the set $e = \{\bm{A}, \bm{U}, \alpha_{\bm{A}}, \alpha_{\bm{U}}, \alpha_{\bm{A}} \bm{A}, \alpha_{\bm{U}}\bm{U} \}$} by using the utility function \textsc{Hash} (Algorithm \ref{util1}), which works similarly as Algorithm \ref{alg:commit}. Finally, each pair $(\bm{rp}_s, \bm{rp}_s^2), s \in e$ will be stored in another state variable $unverified\_elements$ with the key as $msg.sender$. 


\subsubsection{Verify}
After deadline $ddl_1$ set by $\textit{AA}_{\text{trust}}$ in the first stage \emph{Initiate}, the system enters into the stage \emph{Verify}. In this stage, we need to check that each attribute authority \textit{AA} possesses the knowledge of the exponent $s$ used in the list of \textit{s-pair} $L$.

Every \textit{AA} generates the proof $\pi_{s} := \textsc{NIZK}(\bm{rp}_s, h || h_s)$ using Algorithm \ref{alg:NIZK} for each $s \in e$, and broadcasts these proofs as a list \begin{equation}
    L_{\pi} = \{\pi_{\bm{A}}, \pi_{\bm{U}}, \pi_{\alpha_{\bm{A}}}, \pi_{\alpha_{\bm{U}}}, \pi_{\alpha_{\bm{A}} \cdot \bm{A}}, \pi_{\alpha_{\bm{U}} \cdot \bm{U}}\}
\end{equation} through a blockchain transaction to get them verified. The function \textsc{Prove} (Algorithm \ref{SS3}) from contract $SC_{sys}$ takes input $L_{\pi}$ and processes this verification works. 

The function \textsc{Prove} firstly calls function \textsc{SameRatio} of contract $SC_{util}$ (Algorithm \ref{util2}) similar as Algorithm \ref{alg:SameRatio} to examine the authenticity of the published $\bm{rp}_s$ and $\bm{rp}_s^2$. Afterwards, it computes $h_{tmp} := h || \textsc{COMMIT}(\bm{rp}_s)$, and takes $h_{tmp}$ with verified $\bm{rp}_s$ and provided $\pi_s$ as input to the function \textsc{CheckPoK} of contract $SC_{util}$ (Algorithm \ref{util3}), which works similarly as Algorithm \ref{alg:vrfNIZK} and returns \textit{true} if the given proof $\pi_s$ is valid. Finally, the function \textsc{Prove} will remove the list of \textit{s-pair} $L_{rp}$ from the state variable $unverified\_elements$ and store it in the state variable $verified\_elements$. This indicates that the \textit{AA} possesses knowledge of the exponents for the set of \textit{s-pair}.

\subsubsection{Compute and Generate} \label{C&G}
In this stage, the system will generate the public parameters for the attribute-based encryption in two rounds. We use below notation \textsc{powerMulti}$(A, B)$ for the following algorithm \ref{alg:pMu}:

\begin{algorithm}
\caption{Computing power matrix $A$ by matrix $B$}\label{alg:pMu}
\begin{algorithmic}[1]

\Require group elements $A$ and matrix $s$ have same size $l \times k$

\Function{powerMulti}{$A, s$}

\For {$i \leftarrow 1, l$}
    \For {$j \leftarrow 1, k$}
        \State $B[i,j] \gets A[i,j]^{s[i,j]}$
    \EndFor
\EndFor
\State \Return $B$
\EndFunction

\end{algorithmic}
\end{algorithm}

\textbf{Round 1:} We define the first attribute authority \textit{AA} as participant $P_1$, who broadcasts $( V_1, \theta_{V_1}, V_1' )$ as argument of function \textsc{Compute} in contract $SC_{sys}$ (Algorithm \ref{SS4}). The elements $( V_1, \theta_{V_1}, V_1' )$ are constructed as follows:
\begin{align*}
    V_1 &:=  g_1^{\bm{A}_1}\\
    \theta_{V_1} &:= g_1^{\alpha_{\bm{A}_1}}\\
    V_1' &:= g_1^{ \alpha_{\bm{A}_1} \cdot \bm{A}_1}\\
\end{align*}
And the next participant $P_i, i = 2,3,\dots,n$, generates corresponding elements $(V_i, \theta_{V_i}, V_i')$ using Algorithm \ref{alg:pMu}, and also broadcasts them to the contract $SC_{sys}$:
\begin{align*}
    V_i &:= \textsc{powerMulti}(V_{i-1}, \bm{A}_i)\\
    \theta_{V_i} &:= (\theta_{V_{i-1}})^{\alpha_{\bm{A}_i}}\\
    V_i' &:= \textsc{powerMulti}(V_{i-1}', \alpha_{\bm{A}_i}\bm{A}_i)\\
\end{align*}

Since receiving the elements $(V_1, \theta_1, V_1')$ from first participant $P_1$, function \textsc{Computer} of $SC_{sys}$ checks the validity of each incoming elements $(V_i, \theta_{V_i}, V_i')$ published by $P_i$. 

In the end, we define the last valid $V_i$ as one piece of the public parameter:
\begin{multline*}
     g_1^{\bm{A}} = \textsc{powerMulti}(\textsc{powerMulti}(\dots(\\
    \textsc{powerMulti}(\textsc{powerMulti}(g_1^{\bm{A}_1}, \bm{A}_2), \bm{A}_3)\dots,\bm{A}_n)
\end{multline*}

\textbf{Round 2:} In this round, we also define the first attribute authority \textit{AA} as participant $P_1$, who broadcasts $( W_1, \theta_1', W_1' )$ as argument of the function \textsc{Generate} in contract $SC_{sys}$ (Algorithm \ref{SS5}). The elements $(W_1, \theta_1', W_1')$ are constructed as follows:
\begin{align*}
    W_1 &:= ((g_1^{\bm{A}})^{\intercal})^{\bm{U}_1}\\
    \theta_1' &:= g_1^{\alpha_{\bm{U}_1}}\\
    W_1' &:= ((g_1^{\bm{A}})^{\intercal})^{\alpha_{\bm{U}_1} \bm{U}_1}\\
\end{align*}

And the next participant $P_i$ where $i = 2,3,\dots,n$, generates corresponding elements $(W_i, \theta_{W_i}, W_i')$ using Algorithm \ref{alg:pMu}, and also broadcasts them to the contract $SC_{sys}$ :

\begin{align*}
    W_i &:= \textsc{powerMulti}(W_{i-1}, \bm{U}_i)\\
    \theta_{W,i} &:= (\theta_{W_{i-1}})^{\alpha_{\bm{U},i}}\\
    W_i' &:= \textsc{powerMulti}(W_{i-1}', \alpha_{\bm{U}, i}\bm{U}_i)\\
\end{align*}

Invoked by these transactions, contract $SC_{sys}$ checks the validity of each received elements $(W_i, \theta_{W_i}, W_i')$ and updates the value of $W_i$.

As a result of this procedure, the final piece of the public parameter is defined as the last valid $W_i$ received:
\begin{multline*}
     g_1^{\bm{U}^{\intercal}\bm{A}} = W_n = \textsc{PowerMulti}(\textsc{PowerMulti}(\dots(\\
    \textsc{PowerMulti}(\textsc{PowerMulti}(((g_1^{\bm{A}})^{\intercal})^{\bm{U}_1}, \bm{U}_2), \dots,\bm{U}_n))))^{\intercal}
\end{multline*}

At the end of this stage, each Service User (\textit{SU}) can easily get the global public parameters $PP_{ABE}$ of the attribute-based encryption system based on the published values. 

\begin{equation}
    PP_{ABE} := \{g_1, g_2, g_1^{\bm{A}}, g_1^{\bm{U}^{\intercal}\bm{A}}\}
\end{equation}

\subsection{Authority Setup} \label{sub_sec:auth_setup}
This step consists of 3 stages: \emph{Commit and Reveal}, \emph{Verify}, and \emph{Generate}, and finally outputs another global public parameter $PP_{VC}$ and public key $PK$ for each attribute authority \textit{AA}. 

\subsubsection{Initiate}
The contract $SC_{auth}$ deployed by trusted authority $\textit{AA}_{\text{trust}}$ has 4 main functions, \textsc{Commit}, \textsc{Reveal}, \textsc{Prove} and \textsc{Generate}, which also interacts with utility function $SC_{util}$. Its accessibility is also limited by deadlines $(ddl_1', ddl_2', ddl_3')$ and the authorized list $AAlist$ set by $\textit{AA}_{\text{trust}}$. The pseudo-code of contract $SC_{auth}$ is provided in Appendix \ref{sc-au}.

\subsubsection{Commit and Reveal}
Every \textit{AA} firstly samples a set of elements $e'$: a matrix $\bm{X} \xleftarrow{\$} \mathbb{Z}_p^{(k+1) \times (k+1)}$, a vector $\bm{\tau} \xleftarrow{\$} \mathbb{Z}_p^{k+1}$, two secret elements $\sigma, z \in \mathbb{Z}_p^*$, a scalar $\alpha_{z}$ and a scaled element $\alpha_z \cdot z$. 

Therefore, the \textit{AA} defines a set $e'$ for vector commitment as:
\begin{equation*}
    e' := \{z, \alpha_z, \alpha_z\cdot z\}
\end{equation*} and also designates:
\begin{equation*}
    SK := \{ \bm{X}, \bm{\tau}, \sigma \}
\end{equation*} as the secret key. Then, \textit{AA} computes a set of \textit{s-pair} for each element $s$ in both $e'$ and $SK$. We refer to the \textit{s-pair} in $\mathbb{G}_1$ by $\bm{rp}_s$, and the \textit{s-pair} in $\mathbb{G}_2$ by $\bm{rp}_s^2$ as Definition \ref{def:spair}. These \textit{s-pair} are defined as follows:

\begin{itemize}
    \item For matrix $\bm{X}^{\intercal}: \bm{rp}_{\bm{X}} := (g_1^{\bm{A}}, g_1^{\bm{X}^{\intercal} \cdot \bm{A}})$
    \item For vector $\bm{\tau}^{\intercal}: \bm{rp}_{\bm{\tau}} := (g_1^{\bm{A}}, g_1^{\bm{\tau}^{\intercal}\bm{A}})$
    \item For element $\sigma: \bm{rp}^2_{\sigma} := (g_2, g_2^{\sigma})$ 

    \item For element $z: (\bm{rp}_z, \bm{rp}^2_z) = (g, g^z)$
    \item For scalar $\alpha_z$: $(\bm{rp}_{\alpha_z}, \bm{rp}^2_{\alpha_z}) = (g, g^{\alpha_z})$
    \item For scaled element $\alpha_z z$: $(\bm{rp}_{\alpha_z z}, \bm{rp}^2_{\alpha_z z}) = (g, g^{\alpha_z \cdot z})$
\end{itemize}

After that, \textit{AA} computes $h'$ using Algorithm \ref{alg:commit} as follows:
\begin{equation}
    \begin{gathered}
        h_{\bm{X}} := \textsc{COMMIT}(\bm{rp}_{\bm{X}}) \quad h_{\bm{\tau}} := \textsc{COMMIT}(\bm{rp}_{\bm{\tau}})\\
        h_{\sigma} := \textsc{COMMIT}(\bm{rp}^2_{\sigma}) \quad h_{z} := \textsc{COMMIT}(\bm{rp}_{z} || \bm{rp}_z^2)\\
        h_{\alpha_z} := \textsc{COMMIT}(\bm{rp}_{\alpha_z} || \bm{rp}_{\alpha_z}^2) \\
        h_{\alpha_z \cdot z} := \textsc{COMMIT}(\bm{rp}_{\alpha_z \cdot z} || \bm{rp}_{\alpha_z \cdot z}^2)\\
        h' := \textsc{COMMIT}(h_{\bm{X}} || h_{\bm{\tau}} || h_{\sigma} || h_{z} || h_{\alpha_z} || h_{\alpha_z \cdot z})
    \end{gathered}
\end{equation} and broadcasts it to the contract $SC_{auth}$ through blockchain transaction as the argument of the function \textsc{Commit} (Algorithm \ref{AU1}), which will store the value $h'$ into state variable $h\_collector$ if the transaction is valid. Apart from $h\_collector$, several other state variables are defined in the contract $SC_{auth}$, described as follows:

\begin{enumerate}

    \item $h\_collector$ (State Variable): A mapping collection from the blockchain address belonged to one attribute authority $\bm{AA}$ to its commitment $h'$
    \item $unverified\_sk$ (State Variable): A mapping collection from the blockchain address belonged to one attribute authority $\bm{AA}$ to its list of unverified elements $SK$
    \item $unverified\_elements$ (State Variable): A mapping collection from the blockchain address belonged to one attribute authority $\bm{AA}$ to its list of unverified group elements $e'$ in both $\mathbb{G}_1$ and $\mathbb{G}_2$
    
\end{enumerate}  

After $h'$ recorded by contract $SC_{auth}$, \textit{AA} needs to reveal each committed element by passing two lists of \textit{s-pair} \begin{align}
    L_{sk} &= \{\bm{rp}_{\bm{X}}, \bm{rp}_{\bm{\tau}}, \bm{rp}^2_{\sigma}\}\\
    L_{e'} &= \{(\bm{rp}_s, \bm{rp}_s^2) | s \in e'\}
\end{align} as arguments of the function \textsc{Reveal} (Algorithm \ref{AU2}), which computes the hash result of $L_{sk}$ and $L_{vc}$ using utility contract $SC_{util}$ (Algorithm \ref{util1}), and compares the result with the value stored in state variable $h\_collector$. Finally, the valid set of \textit{s-pair} will be stored into state variables $unverified\_elements$ and $unverified\_sk$ of the contract $SC_{auth}$ respectively. 

\subsubsection{Verify}
The system enters into the stage \emph{Verify} after $ddl_1'$ set by trusted authority $\textit{AA}_{\text{trust}}$. 

First of all, attribute authority \textit{AA} generates the proof $\pi_{s} := \textsc{NIZK}(\bm{rp}_s, h || h_s)$ for each $s$ in both $e'$ and $SK$, and broadcast these proofs $\{\pi\}$ as a list \begin{equation}
    L_{\pi}' = \{\pi_{z}, \pi_{\alpha_z}, \pi_{\alpha_z z}, \pi_{\bm{X}}, \pi_{\bm{\tau}}, \pi_{\bm{\sigma}} \}
\end{equation} through transaction before deadline $ddl_2'$. The function \textsc{Prove} of contract $SC_{auth}$ (Algorithm \ref{AU3}) takes input $L_{\pi}'$ as the argument. It checks the validity of these proofs by using utility functions \textsc{SameRatio} and \textsc{CheckPoK} of contract $SC_{utils}$ (Algorithm \ref{util2} and \ref{util3}). 

We assign each \textit{AA} an index $i \in [n-1]$, based on the order in which $SC_{auth}$ receives the complete set of valid published proofs $\{\pi\}$. The trusted authority, denoted as $\textit{AA}_{\text{trust}}$, is assigned the index $n$.

At the end of this stage, we have the verified elements $PK_i = (g_1^{\bm{X}_i^{\intercal}\bm{A}}, \hat{e}(g_1,g_2)^{\bm{\tau}_i^{\intercal} \bm{A}}, g_2^{\sigma})$ and $\{o_i := g_1^{z_i},  g_1^{\alpha_{z_i}}, g_1^{ \alpha_{z_i} z_i}\}$ for each $\textit{AA}_i$.

\subsubsection{Generate}
In the last stage \emph{Generate}, $\textit{AA}_i$ needs to generate a set of group elements $(O, \theta_O, O')$, selects a reasonable number of supported attributes $l_i$, and broadcasts them to contract $SC_{auth}$ before the deadline $ddl_3'$.

The elements $(O, \theta_O, O')$ provided by $\textit{AA}_i$ are constructed as follows:
\begin{align*}
    O_i &:= o_{i,j} = \{(o_j)^{z_i}\}_{i \neq j, j \in [n]}\\
    \theta_{O_i} &:= \{(g_1^{\alpha_{z_j}})^{\alpha_{z_i}}\}_{i \neq j, j \in [n]}\\
    O'_i &:= \{(g_1^{\alpha_{z_j}z_j })^{\alpha_{z_i}z_i}\}_{i \neq j, j \in [n]}\\
\end{align*}

Invoked by this transaction, function \textsc{Generate} of $SC_{auth}$ (Algorithm \ref{AU4}) firstly checks the validity of elements $(O_i, \theta_{O_i}, O'_i)$ and the value of $l_i$, and then records $\textit{AA}_i$'s blockchain address $addr$ with the claimed attribute size $l_i$. For example, if $\textit{AA}_i$ owns the set of attributes $\{entry, mid, senior, agent, manager\}$, the value of attribute size $l_i$ is $5$.

\begin{table}
\caption{A example of address-\textit{AA}-attribute vector Mapping Table}
\label{tab:map tab}
\resizebox{\columnwidth}{!}{%
\begin{tabular}{|l|ll|l|l|ll|l|}
\hline
Address &
  \multicolumn{2}{l|}{$addr_1$} &  $addr_2$ &  ...... &  \multicolumn{2}{l|}{$addr_{n-1}$} & $addr_n$\\ \hline
Attribute Authority &
  \multicolumn{2}{l|}{$\textit{AA}_1$} &  $\textit{AA}_2$ &  ...... &  \multicolumn{2}{l|}{$\textit{AA}_{n-1}$} & $\textit{AA}_{\text{trust}}$\\ \hline
Attribute Representation &  
  \multicolumn{1}{l|}{$v_1$} &
  $v_2$ &
  $v_3$ &
  ...... &
  \multicolumn{1}{l|}{$v_{l-1}$} &
  $v_l$ &
  $v_{l+1}$\\ \hline

\end{tabular}
}
\end{table}

After the contract $SC_{auth}$ receives all the pairs $\{o_i, o_{i,j}, l_i\}$ from $\textit{AA}_i \in \{\textit{AA}_1, \textit{AA}_2, ..., \textit{AA}_n\}$, every Service User (\textit{SU}) can get the global parameter for the vector commitment system
\begin{equation}
    PP_{VC} := \{ g_1, g_2, \{o_i\}_{i \in [n]}, \{o_{i,j}\}_{i, j \in [n], i\neq j} \}
\end{equation}
and generate a mapping table that maps each blockchain address of $\textit{AA}_i$ to its corresponding information as shown in Table \ref{tab:map tab}. We use $l$ to represent the size of complete supported attributes set $\mathcal{X}$ and $v$ to represent the owned attributes for each \textit{AA}.

\subsection{Data User Registration}

To get the global identifier $GID$ which will be used in the process of \textit{Key Generation}, the data user (\textit{DU}) needs to register the owned address $addr_{DU}$ by calling the function \textsc{user\_register} of the contract $SC_{REG}$ (Algorithm \ref{REG}).

\textit{DU} needs to send predefined amount of \textit{GWEI} to the contract $SC_{REG}$ as the registration fee payment. This amount defaults to 1000000 \textit{GWEI}, which is approximately equivalent to 1.63 USD as of September 2023. In return, \textit{DU} receives a value $GID$, which is the hash value of \textit{DU}'s address $addr_{DU}$, also known as the $msg.sender$ of this transaction call. 

Following that, \textit{DU} establishes a secure channel or employs some off-chain methods with $\textit{AA}_i$ that have the required attributes and can verify \textit{DU}'s identity. We assume that \textit{DU} is an agent for one insurance company, and the company itself runs the consensus node of $\textit{AA}_i$ in this system. Therefore, \textit{DU} may easily get verified by showing an ID badge to the person who manages the $\textit{AA}_i$. \textit{DU} then requests that $\textit{AA}_i$ issues a set of attributes $\mathcal{R}_{i,GID}$ in regards to \textit{DU}'s identity. The format of a set of attributes might look like this: $(entry, N/A, N/A, agent, N/A)$ out of the full set of attributes $(entry, mid, senior, agent, manager)$.

For those $\textit{AA}_j, j \neq i$ that do not contain the needed attributes or cannot verify \textit{DU}'s identity, \textit{DU} may just set $\mathcal{R}_{j,GID}$ to be $(N/A, N/A, N/A)$ with $l_j = 3$.

Finally, \textit{DU} receives $\mathcal{R}_{i,GID}$ from $\textit{AA}_i$, sets $\mathcal{R}_{j,GID}$ for $\textit{AA}_j, j\neq i$, and combines them as the set of attributes $\mathcal{R}_{GID}$ which will be used in the process of Key Generation.

\subsection{Key Generation} 
Without loss of generality, we suppose that there are a total set of attributes $\mathcal{X}$, indexed from 1 to $l$ and a total set of attributes authorities $\mathcal{U}$ including $\textit{AA}_{\text{trust}}$, indexed from 1 to $n$. Assume that the attribute authority $\textit{AA}_i$ has a subset of attributes $\mathcal{S}_i$, then we have $\mathcal{S}_i \cap \mathcal{S}_j = \emptyset$ for $i \neq j$ and $i, j \in |n|$ and $\mathcal{S}_1 \cup \mathcal{S}_2 ... \cup \mathcal{S}_n = \mathcal{X}$. 

To get the secret key $sk_{i, GID, \bm{C}}$ which is comprised of multiple key parts $\{K_{j, GID, \bm{C}}\}_{j \in \mathcal{S}_i}$ from $\textit{AA}_i$, Data User (\textit{DU}) must initially generate the attribute vector $\bm{v}_{GID}$. This is based on $\mathcal{R}_{GID}$ that was acquired during the \textbf{Data User Registration} process.

Given that the \textit{DU}'s set of attributes $\mathcal{R}_{i,GID} \in \mathcal{X}$ and the mapping Table \ref{tab:map tab} generated from the process \textbf{Authority Setup}, the attribute vector $\bm{v}_{GID}$ is set as follows:
\begin{enumerate}
    \item Set the first $l$ entries such that $v_k$ = 
        $\begin{cases}
        1 & i \in R_{GID}\\
        0 & i \notin R_{GID}\\
        \end{cases}$
    \item Set the $l+1$ entry to be 1. ($\textit{AA}_{\text{trust}}$ is responsible for this entry)
\end{enumerate}

Then, \textit{DU} randomly chooses $r_1, r_2, ... , r_n \xleftarrow{\$} \mathbb{Z}_p^*$ for each \textit{AA} and combines the arguments $r_i$ with the attribute vector $\bm{v}_{i,GID}$ (represented as a bit string) to produce a committed value $m_i := \textsc{COMMIT}(v_{i,1} || v_{i,2} \dots v_{i,j} || r_i)$, where $i \in [n], j \in [|\mathcal{S}_i|]$.

\begin{table}[]
\caption{A example of $m$ for Vector Commitment }
\label{tab:att vc}
\resizebox{\columnwidth}{!}{%
\begin{tabular}{|l|lll|ll|}
\hline
Attribute Authority & \multicolumn{3}{l|}{$\textit{AA}_1$}                                      & \multicolumn{2}{l|}{$\textit{AA}_2$}          \\ \hline
Attributes          & \multicolumn{1}{l|}{entry} & \multicolumn{1}{l|}{mid}   & senior & \multicolumn{1}{l|}{agent} & manager \\ \hline
Vector Element      & \multicolumn{1}{l|}{$v_1$} & \multicolumn{1}{l|}{$v_2$} & $v_3$  & \multicolumn{1}{l|}{$v_4$} & $v_5$   \\ \hline
Element Value       & \multicolumn{1}{l|}{1}     & \multicolumn{1}{l|}{0}     & 0      & \multicolumn{1}{l|}{1}     & 0       \\ \hline
nonce               & \multicolumn{3}{l|}{$r_1$}                                       & \multicolumn{2}{l|}{$r_2$}           \\ \hline
$m_i$ & \multicolumn{3}{l|}{\textsc{COMMIT}($v_1 || v_2 || v_3 || r_1$)} & \multicolumn{2}{l|}{\textsc{COMMIT}($v_4 || v_5 || r_2$)} \\ \hline
\end{tabular}%
}
\end{table}

In Table \ref{tab:att vc}, for instance, we have two attribute authorities $\textit{AA}_1$ and $\textit{AA}_2$ that possess attributes (entry, mid, senior) and (agent, manager) respectively. If there is a data user \textit{DU} with a set of attributes $\mathcal{R} =$(entry, agent), \textit{DU}'s attribute vector is $\bm{v} = (1,0,0,1,0)$. For $\textit{AA}_1$ and $\textit{AA}_2$, \textit{DU} samples two random values $r_1, r_2 \xleftarrow{\$} \mathbb{Z}_p^*$ and then computes auxiliary information $\bm{aux} = (m_1, m_2)$ using \textsc{COMMIT} (Algorithm \ref{alg:commit}). 

In general, based on the public parameter for vector commitment $PP_{VC} = \{g_1, g_2, \{o_i\}_{i \in [n]}, \{o_{i,j}\}_{i,j \in [n], i\neq j} \}$, generated in \textbf{Authority Setup}, \textit{DU} can calculate the commitment $\bm{C}$ on the attribute vector $\bm{v}$ from its $\mathcal{R}$:

\begin{gather*}
    \bm{C} := o_1^{m_1} o_2^{m_2} \dots o_n^{m_n}    
\end{gather*}

To request the secret key part $K_{j, GID, \bm{C}}, j \in \mathcal{S}_i$, \textit{DU} establishes another secure channel with $\textit{AA}_i$ and sends commitment $C$ along with an opening $op_i$ and nonce $r_i$. Such opening $op_i$ is calculated as follows:

\begin{equation*}
    op_i = \prod^{n}_{j = 1, j \neq i} o_{i,j}^{m_j} = (\prod^{n}_{j = 1, j \neq i} o_j^{m_j})^{z_i}
\end{equation*}

Based on the \textit{DU}'s $GID$, $\textit{AA}_i$ firstly retrieves the information of the set of attributes $\mathcal{R}_{i, GID}$ which have been issued in the previous process \textbf{Data User Registration} and then verifies commitment $\bm{C}$ using the opening $op_i$ and nonce $r_i$ received
\begin{equation*}
    \hat{e}(\bm{C} / o_i^{m_i'}, g_2^{z_i}) \stackrel{?}{=}\hat{e}(op_i, g_2)
\end{equation*}

where $m_i'$ is the value calculated by the $\mathcal{R}_{i,GID}$ issued to DU. If the above check passes, $\textit{AA}_i$ uses a pre-defined random oracle $\mathcal{H}: \mathbb{G}_2 \times \{0,1\}^\lambda \times \mathbb{Z}_p^{k+1} \rightarrow \mathbb{Z}_p^{k+1}$ to generate masking value $\bm{\mu_i} \in \mathbb{Z}_p^*$
\begin{equation*}
    \bm{\mu_i} = \sum_{j = 1}^{i - 1} \mathcal{H}(y_j^{\sigma_i}, GID, \bm{C}) - \sum_{j = i + 1}^{n} \mathcal{H}(y_j^{\sigma_i}, GID, \bm{C})
\end{equation*}

and hash functions $H_1(GID, \bm{C}), ... , H_{k+1}(GID, \bm{C})$ to generate $g_2^{\bm{h}}$ where $\bm{h} \in \mathbb{Z}_p^{k+1}$
\begin{equation*}
    \bm{h}: H(GID, \bm{C}) = (H_1(GID, \bm{C}), ... , H_{k+1}(GID, \bm{C}))^{\intercal}
\end{equation*}

Finally, $\textit{AA}_i$ computes the secret keys \begin{equation}
    sk_{i, GID, \bm{C}} := \{K_{j, GID, \bm{C}}\}_{j \in \mathcal{S}_i}
\end{equation}, which consists of key parts
\begin{equation}
     K_{j, GID, \bm{C}} := g_2^{\bm{\tau_i} - v_j\bm{X}_i\bm{h} + \bm{\mu_i}}
\end{equation}
for each possessed attributes by $\textit{AA}_i$ and send $sk_{i, GID, \bm{C}}$ back to the \textit{DU} through the secure channel.

To get the special secret key part $sk_{n, GID, \bm{C}}$ for the $l+1$ entry, \textit{DU} also needs to communicate with trusted authority $\textit{AA}_{\text{trust}}$ and provides the commitment $\bm{C}$ with the opening $op_{n}$ and nonce $r_{n}$ through the secure channel. As shown in the Table \ref{tab:map tab}, $\textit{AA}_{\text{trust}}$ sets the $m_{n}'$ to be \textsc{COMMIT}($(v_{l+1}|| r_{n})$) where $v_{l+1} = 1$ and then checks the following equation \begin{equation*}
    \hat{e}(\bm{C} / o_{n}^{m_{n}'}, g_2^{z_{n}}) \stackrel{?}{=}\hat{e}(op_{n}, g_2)
\end{equation*}

If it passes, $\textit{AA}_{\text{trust}}$ computes the secret key part similar as $sk_{i, GID, \bm{C}}$
\begin{equation}
    sk_{n, GID, \bm{C}} := K_{l+1} = g_2^{\bm{\tau_{n}} - v_{n}\bm{X}_{n}\bm{h} + \bm{\mu_{n}}}
\end{equation}
and sends it back to \textit{DU}. 

Upon receiving all the responses from each $\textit{AA}_i, i \in [n]$, \textit{DU} will finally gets a complete set of secret keys $\{sk_{i, GID, \bm{C}}\}_{i \in [n]}$.

\subsection{Encryption and Upload}

Given that Attribute-based Encryption (ABE) is significantly more expensive than symmetric key encryption\cite{wang2014performance}, the files that the data owner (\textit{DO}) wants to share are not directly encrypted with ABE. Instead, hybrid encryption of ABE and AES is used for efficiency.

First of all, \textit{DO} randomly samples an AES key $AK$ from the key space and encrypts the file $F$ to be the ciphertext $CT_F$. 

Then, \textit{DO} uploads the ciphertext $CT_F$ to IPFS and records the file location $loc$ returned by IPFS. The metadata message $M$ can then be constructed as: 
\begin{equation*}
    M := (K, loc)
\end{equation*}

Using the policy vector $\bm{x}$ discussed above, \textit{DO} samples a random vector $\bm{s} \in \mathbb{Z}_p^k$, generates a policy vector $\bm{x} := (x_1, x_2, ..., x_n) \in \mathbb{Z}_p^n$ acting as the ciphertext policy, and outputs the ciphertext $CT_M$ consisting of
\begin{gather}
    ct_0 = g_1^{\bm{As}} \\
    ct_i = g_1^{(x_i\bm{U}^{\intercal}+\bm{X}_i^{\intercal})\bm{As}} \\
    ct' =  M \cdot \hat{e}(g_1,g_2)^{\bm{\tau}^{\intercal}\bm{As}}
\end{gather}, where $\bm{\tau} = \sum_{i = 1}^{n} \bm{\tau}_i$.

Lastly, the ciphertext $CT_M$ is sent to the contract $SC_{log}$ with the optional keyword $kw$ that may ease the data retrieval process. The contract $SC_{log}$ will emit this new uploading information to the subscribers as follows: 

\begin{algorithm}
\caption{Data Sharing}\label{Log}
\begin{algorithmic}[1]

\Require $msg.sender == DO$

\Function{log}{$ct, kw$}

    \State $Log.add(ct,kw)$
    \State emit Log$(ct, kw)$
\EndFunction

\Statex
\Function{get}{$index$}

    \If{$index == -1$}
        \State \Return $Log$
    \EndIf
    \If{$index$ is valid}
        \State \Return $Log[index]$
    \EndIf
\EndFunction

\end{algorithmic}
\end{algorithm}

\subsection{Download and Decryption}
As Algorithm \ref{Log} shows, all the encrypted metadata $CT_M$ will be recorded sequentially. If the Data User (\textit{DU}) is interested in one of \textit{DO}'s files, \textit{DU} may subscribe to the event created by the contract $SC_{log}$ and call the Function \textsc{GET} (Algorithm \ref{Log}) to obtain the encrypted metadata $CT_M$.  

To decrypt the ciphertext $CT_M$, \textit{DU} computes
\begin{equation*}
    \hat{e}(ct_0, \prod_{j=1}^{l + 1} K_j) \cdot \hat{e}(\prod ct_i^{v_i}, \bm{h}) = \hat{e}(g_1, g_2)^{\bm{\tau^{\intercal} A s}}
\end{equation*}
    
and tries to recover the message 

\begin{equation}
    ct' / \hat{e}(g_1, g_2)^{\bm{\tau^{\intercal} A s}} = M'
\end{equation}

If \textit{DU}'s attribute vector $\bm{v}$ satisfies the policy vector $\bm{x}$ selected by \textit{DO}, \textit{DU} can retrieve the AES key $AK$ and file location $loc$ from the metadata $M$. Finally, \textit{DU} downloads the encrypted file $CT_F$ from the IPFS based on the location $loc$, and uses $AK$ to decrypt $CT_F$ to recover the original file $F$.

\textit{Correctness.} Since $CT_M = (ct_0, \{ct_i\}^{n}_{i=1}, ct')$, $\{sk_{i, GID, \bm{C}} = \{K_j\}_{j \in \mathcal{R}_i} \}$ and $n = l + 1$, we can compute

\begin{align*}
&\hat{e}(ct_0, \prod_{j=1}^{l + 1} K_j) \cdot \hat{e}(\prod_{i=1}^{n} ct_i^{v_i}, \bm{H}(GID, \bm{C}))\\
&= \hat{e}(g_1^{\bm{A s}}, g_2^{\sum^n_{i=1} \bm{\tau}_i - v_i \bm{X_i h + \mu_i}})\\
&    \cdot \hat{e}(g_1^{\sum^n_{i=1} v_i (x_i \bm{U}^{\intercal} + \bm{X}_i^{\intercal}) \bm{A s}}, g_2^{\bm{h}})\\
&= \hat{e}(g_1,g_2)^{\bm{\tau}^{\intercal}\bm{A s} - \sum^{n}_{i=1} v_i\bm{h}^{\intercal} \bm{X}_i^{\intercal} \bm{A s}} \\
&    \cdot \hat{e}(g_1, g_2)^{\bm{<x,v>h^{\intercal} U^{\intercal} A s} + \sum^{n}_{i=1}v_i\bm{h^{\intercal} X_i^{\intercal} A s}} \\
&= \hat{e}(g_1, g_2)^{\bm{\tau}^{\intercal} \bm{A s}} \cdot \hat{e}(g_1, g_2)^{\bm{<x,v>h^{\intercal} U^{\intercal} A s}}
\end{align*}

If $<\bm{x,v}> = 0$, we obtain $\hat{e}(g_1, g_2)^{\bm{\tau^{\intercal} A s}}$ and can recover the message.

\section{System Analysis}\label{sec:analysis}
\subsection{Performance Analysis}
In this section, we compare our proposed Blockchain-enabled data governance system with previously discussed related works \cite{wang2018blockchain, qin2021blockchain, nishide2008attribute, li2022attribute, gao2020trustaccess, cui2018efficient,zhang2018Security, yang2018improving, zhao2022Secure} in terms of performance and security. In Table \ref{tab:comparison_part2}, we position these aforementioned works that are closely related to our work, providing comprehensive comparisons across the following assessment criteria:

\begin{table}[H]
\centering
\caption{Summary of access control system using Attribute-based Encryption}
\label{tab:comparison_part2}
\resizebox{\columnwidth}{!}{%
\begin{tabular}{l|l|l|l|l}
\hline
\textbf{Approach} &
  \textbf{Group} &
  \textbf{Security} &
  \textbf{Encryption} &
  \textbf{Decryption} \\ \hline
\cite{wang2018blockchain}   & Prime     & S-Rom & $3$ exp    & $2$ pair    \\ \hline
\cite{qin2021blockchain}    & Composite & F-Rom & $(5I+1)$ exp+pair     & $2I$pair+$I$exp  \\ \hline
\cite{yang2018improving}    & Composite & F-ROM & $(7I+1)$exp & $2I$pair + $4I$exp  \\ \hline
\cite{nishide2008attribute} & Prime     & S-STM & $(3n+1)$pair*         & $(3n+1)$pair*  \\ \hline
\cite{gao2020trustaccess}   & Composite & F-STM & $(l+2)$exp & $(l+1)$pair  \\ \hline
\cite{cui2018efficient}     & Prime     & S-STM & $(8I+2)$exp & $(6I+1)$pair + $I$exp  \\ \hline
\cite{zhang2018Security}    & Composite & F-STM & $(6I + 4)$exp*        & $(I+2)$pair + $2I$exp* \\ \hline
\cite{li2022attribute}      & Prime     & S-STM & $(5I+2)$exp & $(3I+1)$pair + $I$exp  \\ \hline
\cite{zhao2022Secure}       & Prime     & S-STM & $5I$pair + 4exp + 2pm & $2I$pair + $3I$exp  \\ \hline
\rowcolor{lightgray}This work                  & Prime     & S-ROM & $(2n+1)$exp & 2 pair + $2n$ exp  \\ \hline
\end{tabular}%
}
Several works did not present the computational complexity information. In such instances, we have either extrapolated the complexity on our own or referenced results presented in the survey by Zhang et al. \cite{zhang2020attribute}. The latter are highlighted with $^*$.
\end{table}

\begin{enumerate}
    \item Group: There are two types of groups used in CP-ABE: prime-order and composite-order groups. It is noted that the design of prime-order CP-ABE is more efficient than composite-order CP-ABE but is more challenging to achieve full security. 
    
    \item Security Model: Standard model (STM) and random oracle model (ROM) are two typical types of security models considered in the CP-ABE scheme. And adversaries are categorized into selective adversaries and adaptive adversaries. If a CP-ABE scheme is secure against adaptive adversaries under the standard model, we denote this scheme as F-STM achieving full security. Likewise, S-ROM represents a CP-ABE scheme robust against selective adversaries under the random oracle model, and F-ROM if it is secure against adaptive adversaries under the random oracle model.
    
    \item Computation Cost: This assessment considers both the encryption and decryption costs in terms of their complexity, measured in terms of certain standard (cryptographic) operations. The following notations are used:
        \begin{enumerate}
            \item[--] $exp$: exponentiation operation
            \item[--] $pair$: bilinear pair operation
            \item[--] $pm$: point multiplication operation
            \item[--] $I$: the access policy complexity for Linear Secret Sharing Scheme (LSSS)
            \item[--] $l$: the size of attributes
            \item[--] $n$: the size of attribute authorities
        \end{enumerate}
    
\end{enumerate}

In terms of computational costs, bilinear pairing (pair) is the most expensive operation compared to exponentiation (exp) and point multiplication (pm), and the LSSS is a special matrix whose rows are labeled by attributes such that the cost of $I$ might be similar as $n$ or $l$. Consequently, our scheme is superior to most of these works with respect to encryption and decryption cost, as shown in Table \ref{tab:comparison_part2}, but it is constrained in terms of the conjunction access policy. More precisely, for ciphertext $CT_M$ consisting of 
\begin{gather*}
    ct_0 = g_1^{\bm{As}}\\
    ct_i = g_1^{(x_i\bm{U}^{\intercal}+\bm{X}_i^{\intercal})\bm{As}}\\
    ct' =  M \cdot \prod_{i=1}^{n} (\hat{e}(g_1,g_2)^{\bm{\tau}_i^{\intercal}\bm{A}})^s
\end{gather*}, $ct_0$ costs $1$exp, $ct_i$ costs $n$exp, and $ct'$ costs another $n$exp as the result $\hat{e}(g_1,g_2)^{\bm{\tau}_i^{\intercal}\bm{A}}$ is given by each public key $PK_i$. Similarly, for decrypting ciphertext $CT_M$, the formula $\hat{e}(ct_0, \prod_{j=1}^{l+1} K_i) \cdot \hat{e}(\prod ct_i^{v_i}, \bm{h})$ costs $2$pair + $2n$exp.


\subsection{Policy Hiding}

Policy-hiding means that the ciphertext policy is hidden from inspection. In our approach, we achieve a weaker concept known as weakly attribute-hiding, which ensures that the policy remains unknown to all parties except those who can decrypt the ciphertext. 
Our access control system is constructed on top of the decentralized inner-product predicate encryption scheme in \cite{michalevsky2018decentralized}. Combined with special policy encoding, the proposed construction in Michalevsky's work is proved to be weakly attribute-hiding against chosen plaintext-attack under Assumption \ref{asp:sp-klin} in the presence of corrupt authorities. 

\cite{michalevsky2018decentralized} also provides detailed proof that, in the absence of corrupted authorities, the advantage of a PPT adversary $\mathcal{A}$ in winning a sequence of games, beginning with the actual scheme and ending with a challenge ciphertext, is negligible in the security parameter $\lambda$ against a challenger $S$.

However, the privacy of attributes can not be maintained if adversary $\mathcal{A}$ colludes with a corrupt authority and asks it to provide it with an attribute key $K_j, j \in [l+1]$ for a special value $v_j \in \mathcal{S}_i$. \begin{equation*}
     K_j = g_2^{\bm{\tau_i} - v_i\bm{X}_i\bm{h} + \bm{\mu_i}}
\end{equation*} 

Let $\mathcal{R}$ represent a subset of attributes included in the ciphertext policy. One of the naive constructions of policy vector $\bm{x}$ might be as follows:

\begin{enumerate}
    \item Set the first $l$ entries such that $x_i$ = 
        $\begin{cases}
        1 & i \in \mathcal{R}\\
        0 & i \notin \mathcal{R}\\
        \end{cases}$
    \item Set the $l+1$ entry to $-|\mathcal{R}|$.
\end{enumerate}

If \textit{DO} sets the policy vector to be $\bm{x} = (1,1,0,0,0,-2)$ and \textit{DU} has the attribute vector to be $\bm{v} = (1,1,0,1,0,1)$, it's easily to get $\bm{x} \cdot \bm{v} = 0$, which is the precondition of the successful decryption. A trustworthy authority should only issue keys for values $v_i = 0$ or $v_i = 1$. A corrupted authority, on the other hand, can provide an adversary with a key issued for a specific value $v_i$, which may "satisfy" the policy vector despite lacking all necessary attributes. If the policy vector is still to be $\bm{x} = (1,1,0,0,0,-2)$ and the corrupted \textit{AA} issues the key for the attribute vector $\bm{v} = (2,0,0,0,0,0)$, we can also get $\bm{x} \cdot \bm{v} = 0$. As a result, rather than using $0 ~or~ 1$ to indicate the absence or presence of an attribute, the enhanced scheme in \cite{michalevsky2018decentralized} encodes the required attributes using randomly sampled $r_i$ over a large field, defeating an attempt to craft a key by arbitrarily selecting a value $v_i$ that would result in a zero inner product.

\subsection{Receiver Privacy} 
According to the Definition \ref{def:IPPE}, the receiver must provide its attribute vector $\bm{v}$ to each attribute authority \textit{AA} from which a key is requested. In consequence, \textit{AA} learns not only if the user possesses the attribute that \textit{AA} owns, but also all of the user's other attributes. This appears to violate the privacy of the user in a decentralized setting.

Therefore, Michalevsky and Joye propose an enhancement that provides additional privacy protection for the attribute vector $\bm{v}$ in their work \cite{michalevsky2018decentralized}, which uses it to hide the set of receiver attributes $\bm{v}$ from authorities. This technique is called \textit{position-binding} Vector Commitments, introduced by Catalano and Fiore \cite{catalano2013vector}. Our access control system is based on this work, but it has been slightly modified to work with asymmetric pairings $\hat{e}: \mathbb{G}_1 \times \mathbb{G}_2 \rightarrow \mathbb{G}_T$, which meets two security requirements under Assumption \ref{asp:s-cdh}:
\begin{enumerate}
    \item \textbf{Position Hiding:} Even after seeing some openings for certain positions, an adversary cannot tell whether a commitment was made to a sequence of messages, $(m_1, \cdots, m_n)$ or $(m_1', \cdots, m_n')$.
    \item \textbf{Position Binding:} An adversary should not be able to open a vector commitment to two different messages at the same position.
\end{enumerate}

Since hiding is not a critical property and can be easily achieved in the realization of vector commitments \cite{catalano2013vector}, only the proof of position binding is provided below:
\begin{proof}
Suppose an efficient adversary $\mathcal{A}$ can produce two valid openings $op$ to different messages $(m, m')$ at the same position $i$. In that case, we can build an algorithm $\mathcal{B}$ that uses $\mathcal{A}$ to break the Square Computational Diffie-Hellman assumption.

As we know from Definition \ref{def:VC}, for a sequence of messages $m_1, m_2, \dots, m_n$ and public parameter $pp = (g_1, g_2, \{o_i\}_{i \in [n]}, \{o_{i,j}\}_{i,j \in [n], i \neq j})$, the vector commitment is $\bm{C} = o_1^{m_1}o_2^{m_2}\cdot o_n^{m_n}$ and the opening for position $i$ is $op_i =  (\prod^n_{j = 1, j \neq i} o_j^{m_j})^{z_i}$. 

The efficient algorithm $\mathcal{B}$ takes as input a tuple $(g_1, g_1^r, g_2^r)$ and aims to compute $g_1^{r^2}$ to break Assumption \ref{asp:s-cdh}.

First, $\mathcal{B}$ selects a random position $i \xleftarrow{\$} [n]$ on which adversary $\mathcal{A}$ will break the position binding. Next, $\mathcal{B}$ chooses $z_j \xleftarrow{\$} \mathbb{Z}^*_p$ where $\forall j\in [n], j \neq i$ and then computes:
\begin{gather*}
    \forall j \in [n]\setminus{i}: o_j = g_1^{z_j}, o_{i,j} = (g_1^r)^{z_j}, o_i = g_1^r\\
    \forall k,j \in [n]\setminus{i}, k \neq j: o_{k,j} = g_1^{z_k z_j}
\end{gather*}

Second, $\mathcal{B}$ sets $pp = (g_1, g_2, \{o_i\}_{i \in [n]}, \{o_{i,j}\}_{i,j \in [n], i \neq j})$ and runs $\mathcal{A}(pp)$, that outputs a tuple $(\bm{C}, m, m', j, op_j, op_j')$ such that $m \neq m'$ and both $op_j$ and $op_j'$ are valid.

Finally, $\mathcal{B}$ computes
\begin{equation*}
    g_1^{r^2} = (op_i' / op_i)^{(m_i - m_i')^{-1}}
\end{equation*}

Since openings verify the two messages $(m_i, m_i')$ correctly at position $i$, then it holds:
\begin{equation*}
    \hat{e}(\bm{C}, g_2^r) = \hat{e}(o_i^{m_i}, g_2^r) \hat{e}(op_i, g_2) = \hat{e}(o_i^{m_i'},g_2^r)\hat{e}(op_i', g_2)
\end{equation*} which means that
\begin{equation*}
    \hat{e}(o_i, g_2^r)^{m_i - m_i'} = \hat{e}(op_i'/op_i, g_2)
\end{equation*}
Since $o_i = g_1^r, op_i'/op_i = (g_1^{r^2})^{m_i - m_i'}$,\\ we have:
\begin{gather*}
    \hat{e}(o_i, g_2^r)^{m_i - m_i'} = \hat{e}(g_1, g_2)^{r^2(m_i-m_i')}\\
    \hat{e}(op_i'/op_i, g_2) = \hat{e}((g_1^{r^2})^{m_i - m_i'}, g_2) = \hat{e}(g_1, g_2)^{r^2(m_i-m_i')}
\end{gather*}, which justifies the correctness of $\mathcal{B}$'s output. Therefore, if the Square Computational Diffie-Hellan assumption holds, the scheme satisfies the position-binding property. 
\end{proof}

\subsection{Other Security Requirements}

\subsubsection{Trustability} 

Most existing solutions require an intermediary entity to ensure reliable and secure data management, resulting in expensive costs to prevent a single point of failure and privacy leakage. To overcome these obstacles, our approach employs the blockchain's distribution, decentralization, transparency, and immutability characteristics. By publishing the encrypted metadata, which consists of the AES key $AK$ and file location $loc$, to the blockchain, we can maintain the integrity of access control management without requiring any intermediaries.

\subsubsection{Traceability}

Our system can track and validate access control data on the blockchain. Any activities, including setup, registration, key generation, encryption, and data uploading, are recorded as immutable transactions.


\subsection{Collusion Attack Analysis}
A fundamental requirement for an ABE scheme is to prevent collusion attacks. Let $\textit{DU}_1$ and $\textit{DU}_2$ be two users, possessing sets of secret keys $sk_1$ and $sk_2$, corresponding to the attribute vector $\bm{v} = (v_1, v_2, \cdots, v_l, v_{l+1})$. $sk_1$ consists of key-parts $\{K_{1,j}\}_{j \in [l+1]}$ that enable obtaining a secret related to attribute $v$ from every attribute authorities $\textit{AA}_i$ for $\textit{DU}_1$'s processed attributes element $v \in \mathcal{R}_1$. $sk_2$ consists of key-parts $\{K_{2,i}\}_{i \in [l+1]}$that enable obtaining a secret related to attribute element $v_j$ from every attribute authorities $\textit{AA}_i$ for $\textit{DU}_2$'s processed attributes $v \in \mathcal{R}_2$. The collusion prevention is against that $\textit{DU}_1$ and $\textit{DU}_2$ can mix their key-parts $(\{K_{1,j}\}, \{K_{2,j}\})$ in a way that constructs a secret key $sk'$ to a new vector $\bm{v}'$ such that $v \in \mathcal{R}_1 \cup \mathcal{R}_2$. Otherwise, users can collude to decrypt the ciphertext, which is not accessible for either of them. 

Therefore, we propose two mechanisms to restrict key combinations:
\begin{enumerate}
    \item \textbf{Global Identifier} associates each secret key $sk$ with an identity by incorporating it into the key parts $K$ issued by the attribute authorities.
    \item \textbf{Masking Term} maps a combination of the public keys of all other authorities $\{g_2^{\sigma}\}$, global identifier $GID$, and the vector commitment $\bm{C}$ to a random element 
\end{enumerate}

It is worth noting that our scheme necessitates one special authority $\textit{AA}_{n}$ to be trusted specifically to issue keys only for $v_{l+1} \neq 0$. As previously discussed, this design choice ensures our scheme's security, even if some corrupt authorities collude with adversaries to compute a key for any distinct value $v$ within the attribute vector $\bm{v}$.

\subsection{Rogue-key Attack Analysis}
First, we provide proof that the original scheme, as proposed by Michalevsky et al. \cite{michalevsky2018decentralized} and outlined in Definition \ref{def:IPPE}, is vulnerable to a Rogue-key Attack. Following this, we discuss how our blockchain-based data governance mechanism effectively mitigates this vulnerability.

\subsubsection{Attack}\label{rk-attack}
We demonstrate that a corrupt authority can learn the key parts $K_{l+1}$ corresponding to $v_{l+1} \in \bm{v}$ issued by trusted authority $\textit{AA}_{\text{trust}}$ and thus decrypt the ciphertext in Michalevsky and Joys's scheme \cite{michalevsky2018decentralized}, without having the required secret key $sk$ satisfying the attached access policy. This is a typical attack, called Rogue-key Attack, in which the adversary uses a public key, a function of honest users' keys \cite{ristenpart2007power}.

\begin{proof}
First, an adversarial authority $\textit{AA}_{ad}$ holds until all the other attribute authorities $\textit{AA}_i, i \neq ad, i \in n$ publish their public keys.

\begin{equation*}
    PK_i = (g_1^{\bm{X}_i^{\intercal} \bm{A}},  \hat{e}(g_1,g_2)^{\bm{\tau}_i^{\intercal} \bm{A}}, y := g_2^\sigma )
\end{equation*}
and then calculates the public key as follows:
\begin{gather*}
    g_1^{\bm{X}_{ad}^{\intercal} \bm{A}} := - \sum_{j = 0, j \neq ad}^n g_1^{\bm{X}_j^{\intercal} \bm{A}}\\
    \hat{e}(g_1,g_2)^{\bm{\tau}_{ad}^{\intercal} \bm{A}}, y_{ad} := g_2^{\sigma_i}
\end{gather*}

Second, an adversarial data user $\textit{DU}_{ad}$ creates an attribute vector $\bm{v}' = (0,0, \cdots, 0,1)$ and requrests key parts $K_i$ from all the attribute authorities $\textit{AA}_i$. 

As in the encryption phase, data owner \textit{DO} will collect all the public keys $PK_i, i \in [n + 1]$ from each $\textit{AA}_i$ and outputs the cipher text $CT$ consisting of 
\begin{gather*}
    ct_0 = g_1^{\bm{As}} \\
    ct_i = g_1^{(x_i\bm{U}^{\intercal}+\bm{X}_i^{\intercal})\bm{As}} \\
    ct' =  m \cdot \hat{e}(g_1,g_2)^{\bm{\tau}^{\intercal}\bm{As}}
\end{gather*}, where $\bm{\tau} = \sum_{i = 1}^{n} \bm{\tau}_i$

Third, given the ciphertest $(ct_0, \{ct_i\}, ct')$ and received secret keys $\{K_i\}$, adversary $\textit{DU}_{ad}$ decrypts it as following:

\begin{align*}
&\hat{e}(ct_0, \prod_{j=1}^{l+1} K_i) \cdot \hat{e}(\prod_{i=1}^{n} ct_i^{v_i}, \bm{H}(GID, \bm{C}))\\
&= \hat{e}(g_1^{\bm{A s}}, g_2^{\sum^n_{i=1} \bm{\tau}_i - v_i \bm{X}_i \bm{h} + \bm{\mu}_i})\\
&    \cdot \hat{e}(g_1^{\sum^n_{i=1} v_i (x_i \bm{U}^{\intercal} + \bm{X}_i^{\intercal}) \bm{A s}}, g_2^{\bm{h}})\\
&= \hat{e}(g_1,g_2)^{\bm{\tau}^{\intercal}\bm{A s} - \sum^{n}_{i=1} v_i\bm{h}^{\intercal} \bm{X}_i^{\intercal} \bm{A s}} \\
&    \cdot \hat{e}(g_1, g_2)^{\bm{<x,v>h^{\intercal} U^{\intercal} A s} + \sum^{n}_{i=1}v_i\bm{h^{\intercal} X_i^{\intercal} A s}} \\
&= \hat{e}(g_1, g_2)^{\bm{\tau}^{\intercal} \bm{A s}} \cdot \hat{e}(g_1, g_2)^{\bm{<x,v>h^{\intercal} U^{\intercal} A s}}
\end{align*}

Since $\textit{DU}_{ad}$'s attribute vector is $\bm{v}' = (0,0, \cdots, 0,1)$, we have 
\begin{equation} \label{target}
    \hat{e}(g_1, g_2)^{\bm{\tau}^{\intercal} \bm{A s}} \cdot \hat{e}(g_1, g_2^{\bm{h}})^{x_{l+1} \bm{U}^{\intercal} \bm{A s}}
\end{equation} 
In order to decrypt the ciphertext, we need to cancel out the last element in the above Equation \ref{target}. Therefore, adversary $\mathcal{A}$ could use the published $\{ct_i\}$ to calculate attacking component $\omega$.
\begin{align*}
    \omega :&= \prod g_1^{(x_i\bm{U}^{\intercal}+\bm{X}_i^{\intercal})\bm{As}}\\
    &= g_1^{\sum x_i \bm{U}^{\intercal}\bm{As}} \cdot g_1^{\sum (\bm{X}_i^{\intercal} \bm{A})\bm{s}}
\end{align*}
Based on $g_1^{\bm{X}_{ad}^{\intercal} \bm{A}} := - \sum_{j = 0, j \neq ad}^n g_1^{\bm{X}_j^{\intercal} \bm{A}}$, we obtain $\omega = g_1^{\sum_{i=1}^n x_i \bm{X}^{\intercal}\bm{As}}$ and can further cancel the last component in the Equation \ref{target} as $x_{l+1} = - \sum_{i=1}^n x_i$.

In the end, the adversary recovers the message by computing 
\begin{equation*}
    ct' / \hat{e}(g_1, g_2)^{\bm{\tau}^{\intercal} \bm{A s}} = m.
\end{equation*}
\end{proof}

\subsubsection{Proof of security of our approach} \label{rk-pos} 
One way to mitigate rogue-key attack requires proof of knowledge upon public key registration with a certificate authority (CA) \cite{ristenpart2007power}. In the absence of a CA, Non-interactive Zero-knowledge Proof (NIZK) could be utilized to solve this security issue in a decentralized manner, as has been the case for our approach. 

\begin{proof}
In the process of \textbf{Authority Setup}, every \textit{AA} needs to generate a commitment of secrets used for public key $PK$, and then creates its proof of knowledge using \textbf{NIZK} (Algorithm \ref{alg:NIZK}). 

Since NIZK is built upon Schnorr Identification Protocol \cite{schnorr1990efficient}, our scheme is resistant to rogue-key attacks under the Assumption \ref{asp:KCA}, i.e., Knowledge of Coefficient Assumption \cite{bowe2017scalable}. Particularly, given a simple example $(\bm{rp}_s = (A, B = s \cdot A), h)$, adversary $\mathcal{A}$ can not generate a valid NIZK of $B$ without knowing $s$ with non-negligible probability. 

As defined in the NIZK of $s$ (Algorithm \ref{alg:NIZK}), a valid proof is $\pi = (R, u)$ such that
\begin{align*}
    \alpha &\xleftarrow[]{\$} \mathbb{Z}_p^*\\
    R &\gets \alpha \cdot A\\
    u &\gets \alpha + c \cdot s
\end{align*}, where $c$ is deterministic by $R$ and a given string $h$. Thus, if such $(R,u)$ is provided to $\mathcal{A}$ and $s$ used to generate $u$ is unknown, $\mathcal{A}$ can construct a pair $(x, y = u \cdot A - R)$ for a given $(\bm{rp}_s = (A, B = s \cdot A), h)$ and also there exists an efficient deterministic $\mathcal{X}$ that outputs $c$. 

However, because of the Knowledge of Coefficient Assumption \cite{bowe2018multi}, the probability that both 

\begin{enumerate}
    \item $\mathcal{A}$ `succeeds', meaning it satisfies the condition: SameRatio($(A, x),(B,y)$)
    \item $\mathcal{X}$ `fails', meaning $x \neq g_1^{c}$
\end{enumerate}

is negligible. Thus, adversary $\mathcal{A}$ can not generate a valid proof without knowing $s$ based on given information. In other words, $\mathcal{A}$ can not generate a public key $PK$ with a valid NIZK, if and only if $\mathcal{A}$ has the knowledge of $PK$'s exponent $s$. As such, our scheme is secure against the rogue-key attack that malicious \textit{AA} can not register a $PK$ that is a function of others.

\end{proof}

\subsection{Inferring the Secret Vector in Ciphertext}
With its further study of the Decentralized Policy-hiding ABE scheme \cite{michalevsky2018decentralized}, we also identified a potential risk associated with the generation of public parameters during \textbf{Setup}. The detail of the risk and proof of security is described in the following sections. 

\subsubsection{Vulnerability} \label{isv-risk}
As defined in Definition \ref{def:IPPE}, a trusted third-party (TTP) or an attribute authority (\textit{AA}) needs to pick a set of random numbers $a_1, a_2, \dots , a_k \xleftarrow{\$} \mathbb{Z}^*_p$, and then generate a random matrix
$\bm{A} = \begin{pmatrix}
a_1 & 0 & & 0 \\
0   & a_2 & & 0\\
\vdots & \ddots &  & \\
0 & 0 & \cdots & a_k\\
1 & 1 & \cdots & 1 \\
\end{pmatrix} \in \mathbb{Z}^{(k+1)\times k}_p$ with \\
$\bm{a}^{\perp} = \begin{pmatrix}
a_1^{-1}\\
a_2^{-1}\\
\vdots\\
a_k^{-1}\\
-1\\
\end{pmatrix} \in \mathbb{Z}^{(k+1)}_p$, which $\bm{A^{\intercal}a}^{\perp} = 0$ during \textbf{Setup}.\\ 

\begin{proof}
If the above generation is conducted by a PPT adversary $\mathcal{A}$, or if the sensitive information $\bm{a}^{\perp}$ is exposed, $\mathcal{A}$ can infer the value of $g_1^{\bm{s}}$ from any published ciphertext $CT$.

As we know, to generate a ciphertext $CT$, data owner (\textit{DO}) firstly samples a random vector $\bm{s}  = (s_1, s_2, \dots, s_k) \in \mathbb{Z}_p^k$, then creates a policy vector $\bm{x} = (x_1, x_2, ..., x_{n-1}, x_{n}) \in \mathbb{Z}_p^{n}$ acting as the ciphertext policy, and finally outputs the ciphertext $CT$ consisting of 
\begin{gather*}
    ct_0 = g_1^{\bm{As}} \\
    ct_i = g_1^{(x_i\bm{U}^{\intercal}+\bm{X}_i^{\intercal})\bm{As}} \\
    ct' =  M \cdot \hat{e}(g_1,g_2)^{\bm{\tau}^{\intercal}\bm{As}}
\end{gather*}, where $\bm{\tau} = \sum_{i = 1}^{n} \bm{\tau}_i$ and $\hat{e}(g_1,g_2)^{\bm{\tau}_i^{\intercal}}$ collected from published public keys $PK_i$.

Since 
\begin{equation*}
\bm{As} = \begin{pmatrix}
a_1s_1\\
a_2s_2\\
\vdots\\
a_ks_k\\
s_1 + s_2 + \dots + s_k\\
\end{pmatrix}
\end{equation*}

$\mathcal{A}$ can use the result $ct_0 = g_1^{\bm{As}}$ and known $\bm{a}$ to calculate:

\begin{equation*}
    (g_1^{\bm{As}})^{\bm{a}} = g_1^{\bm{Asa}}
\end{equation*}

For the exponent $\bm{Asa}$, we have 
\begin{align*}
\bm{Asa} &= \begin{pmatrix}
a_1s_1\\
a_2s_2\\
\vdots\\
a_ks_k\\
s_1 + s_2 + \dots + s_k\\
\end{pmatrix} \cdot (a_1^{-1}, a_2^{-1}, \dots, a_k^{-1}, -1)\\
&= \begin{pmatrix}
s_1   & a_1a_2^{-1}s_1 &\dots & a_1a_k^{-1}s_1 & -a_1s_1\\
a_1^{-1}a_2s_2   & s_2 &\dots & a_2a_k^{-1}s_2 & -a_2s_2\\
\vdots & \ddots &  & \\
a_1^{-1}a_ks_k   & a_2^{-1}a_ks_k &\dots & s_k & -a_ks_k\\
a_1^{-1}\sum{s} & a_2^{-1}\sum{s} & \cdots & a_k^{-1}\sum{s} & -\sum{s} \\
\end{pmatrix}   
\end{align*}

It is simple to determine that the $i$th elements from exponents $\bm{As}$ and $\bm{a}$ partially cancel each other out, and the remaining element $s_i$ is the targeting exponent. Therefore, adversary $\mathcal{A}$ might extrapolate $g_1^{\bm{s}}$ from the published ciphertext $ct_0 = g_1^{\bm{Asa}}$ with the knowledge of $\bm{a}^{\perp}$.

\end{proof}

\subsubsection{Proof of security with our approach} \label{isv-pos}
One potential mitigation of the inferring attack is to generate a composite public parameter $PP$ by multiple \textit{AA}, such that neither of those individual \textit{AA}s knows the secrets $\bm{a}$. This is achieved relying on the Knowledge of Exponent Assumption (KEA), introduced by Damg{\aa}rd in \cite{damgaard1991towards}, more precisely:

\begin{proof}
From the Section \ref{C&G} \textbf{Round 1}, the first participant $P_1$ samples a random matrix $\bm{A}_1$ and a scalar $\alpha_{\bm{A}_1}$, constracts the elements $(V_1 = g_1^{\bm{A}_1}, \theta_{V_1} = g_1^{\alpha_{\bm{A}_1}}, V_1' = g_1^{\alpha_{\bm{A}_1}\bm{A}_1})$, and publishes these information for the next participant. 

Upon receiving the $(V_1, \theta_{V_1}, V_1')$, $P_2$ samples his own random matrix $\bm{A}_2$ and a scalar $\alpha_{\bm{A},2}$, calculates 

\begin{align*}
    V_2 &:= \textsc{powerMulti}(V_1, \bm{A}_2)\\
    \theta_{V_2} &:= (\theta_{V_1})^{\alpha_{\bm{A}_2}}\\
    V_2' &:= \textsc{powerMulti}(V_1', \alpha_{\bm{A}_2}\bm{A}_2)\\
\end{align*} using Algorithm \ref{alg:pMu}, and also broadcasts them to the next participant. 

Similarly, the next participant performs the same operation until the very last one. In the end, we have a composite
\begin{multline*}
     g_1^{\bm{A}} = \textsc{powerMulti}(\textsc{powerMulti}(\dots(\\
    \textsc{powerMulti}(\textsc{powerMulti}(g_1^{\bm{A}_1}, \bm{A}_2), \bm{A}_3)\dots,\bm{A}_n)
\end{multline*}, and no participant learns the secrets of others unless one must collude with every other participant under Assumption \ref{asp:KCA}. 

Another security concern is the consistency that an adversary with index $i$ can sample different $(\bm{A}_i, \bm{A}_i')$ and $(\alpha_{\bm{A}_i}, \alpha_{\bm{A}_i}')$ to compute $V_i, \theta_{V_i}$ and $V_i'$, rendering final composite invalid and unusable. 

Since a set of group elements 
\begin{equation*}
    e =(\bm{A}, \bm{U}, \alpha_{\bm{A}}, \alpha_{\bm{U}}, \alpha_{\bm{A}} \cdot \bm{A}, \alpha_{\bm{U}} \cdot \bm{U}) \cdot g
\end{equation*} in both $\mathbb{G}_1$ and $\mathbb{G}_2$ has been committed and revealed by each \textit{AA} as described in Section \ref{C&R}, we may validate that adversary's $(V_i, \theta_{V_i}, V_i')$ is a proper multiple of $P_{i-1}$'s components:

\begin{align*}
    \hat{e}(V_i, g_2) &= \hat{e}(V_{i-1}, g_2^{\bm{A}_i})\\
    \hat{e}(\theta_{V_i}, g_2) &= \hat{e}(\theta_{V_{i-1}}, g_2^{\alpha_i})\\
    \hat{e}(V_i', g_2) &= \hat{e}(V_{i-1}', g_2^{\alpha_{\bm{A},i}\bm{A}_i})
\end{align*}

Even if all the other participants have collaborated, this is still a robust and sufficient setup scheme, even if only one participant is honest and does not reveal its secrets. Hence, the greater the number of unrelated participants in a trusted setup, the less likely the possibility of invalid and unusable public parameter $PP_{ABE}$ \cite{wilcox2016design}.

\end{proof}

\section{Concluding Remarks}\label{sec:conclusion}

Securing cloud data access and protecting identity privacy are legitimate concerns for many use cases, which this work addresses with a blockchain-based data governance system that is secure and privacy-preserving. A combination of attribute-based encryption (ABE) and the Advanced Encryption Standard (AES) makes the system efficient and responsive to real-world conditions. Our ABE encryption system can handle multi-authority scenarios while protecting identity privacy and hiding ABE's policy. 

However, because our system is built on top of Michalevsky's scheme\cite{michalevsky2018decentralized} and Bowe's setup \cite{bowe2018multi}, it has inevitably inherited a few drawbacks: First, it only supports fixed-size attributes and authorities, which means that any changes to these may necessitate requesting a new system setup or a new key for \textit{DU} from all authorities. Second, because each \textit{AA}'s public key must be shared with others for the computation of masking terms $\mu_i$, the system requires coordination among authorities during the setup phase. Third, an authorized user should be unable to learn which attributes were included in the encryption policy, which is a desirable property called \textit{strongly attribute-hiding}. However, our system can only achieve \textit{weakly attribute-hiding}, such that the policy is hidden from all entities other than authorized users. Finally, implementing trusted setup to protect against rogue-key attacks further complicates the entire setup process posing limitations to scenarios in which authorities might join and leave frequently. Addressing these shortcomings comprise our planned future work.

\bibliographystyle{unsrt}
\bibliography{APPENDIX/08-ref} 



\section{Appendix} \label{sec:app}
\subsection{System Contracts} \label{sc-sys}
The function \textsc{Commit} and the state variable $h\_collector$ are defined as following:

\begin{enumerate}

    \item $h\_collector$ (State Variable): A mapping collection from the blockchain address belonged to one attribute authority to its commitment $h$
    
\end{enumerate} 

\begin{algorithm}
\caption{Contract $SC_{sys}$: Part 1}\label{SS1}
\begin{algorithmic}[1]

\Function{Commit}{$h$}
    
    \If {$msg.sender \notin AAlist$ OR $block.timestamp > ddl_1$}    \Comment{Requirement check}
        \State throw 
    \EndIf

    \If  {$msg.sender \notin h\_collector$} 
        \Comment{Resubmitting check}

        \State $h\_collector[msg.sender] \gets h$

    \EndIf

\EndFunction

\algstore{part1}
\end{algorithmic}
\end{algorithm}

The function \textsc{Reveal} and some state variables used are defined as follows:

\begin{enumerate}

    \item $h\_collector$ (State Variable): A mapping collection from the blockchain address belonged to one attribute authority to its commitment $h$
    \item $unverified\_elements$ (State Variable): A mapping collection from the blockchain address belonged to one attribute authority to its unverified list of \textit{s-pair}s $L = \{(\bm{rp}_s, \bm{rp}_s^2) | s \in e\}$ in both $\mathbb{G}_1$ and $\mathbb{G}_2$
    
\end{enumerate}  

\begin{algorithm}
\caption{Contract $SC_{sys}$: Part 2}\label{SS2}
\begin{algorithmic}[1]
\algrestore{part1}
\Function{Reveal}{$L_{rp}$}
    
    \If {$msg.sender \notin h\_collector$ OR $block.timestamp > ddl_1$} \Comment{Requirement check}
        \State throw 
    \EndIf

    \ForAll{$(a, b) \in L_{rp}$}
        \State $h \gets h ||$ \Call{Hash}{a} $||$ \Call{Hash}{b}
    \EndFor
    
    \If {$h\_collector(msg.sender) \neq h$}\Comment{Check validity}
        \State throw 
    \EndIf

    \If  {$msg.sender \notin unverified\_elements$} 
        \Comment{Resubmitting check}
        \State $unverified\_elements[msg.sender] \gets L_{rp}$

    \EndIf

\EndFunction

\algstore{part2}
\end{algorithmic}
\end{algorithm}

The function \textsc{Prove} and some state variables used are defined as follows:
\begin{enumerate}
    \item $verified\_elements$ (State Variable): A mapping collection from the blockchain address belonged to one attribute authority to its verified list of \textit{s-pair}s $L = \{(\bm{rp}_s, \bm{rp}_s^2) | s \in e\}$ in both $\mathbb{G}_1$ and $\mathbb{G}_2$

\end{enumerate}

\begin{algorithm}
\caption{Contract $SC_{sys}$: Part 3}\label{SS3}
\begin{algorithmic}[1]
\algrestore{part2}
\algblock{Assembly}{End}
\Function{Prove}{$L_{\pi}$}

    \If{$unverified\_elements[msg.sender] == []$ OR $block.timestamp < ddl_1$ OR $block.timestamp > ddl_2$}
        \State throw 
    \EndIf

    \State $L_{rp} \gets unverified\_elements[msg.sender]$
    \For{$i \leftarrow 0, 5$}
        \State $(rp, rp2) \gets L_{rp}[i]$
        \State $pi \gets L_{\pi}[i]$
        
        \If{not \Call{SameRatio}{$rp,rp2$}}
            \State throw
        \EndIf
        \State $h_{tmp} \gets  h\_collector[msg.sender] || \Call{Hash}{rp_s} $
        \If{not \Call{CheckPoK}{$rp_s, \pi_s, h_{tmp}$}}
            \State throw
        \EndIf
    \EndFor

    \State $verified\_elements[msg.sender] \gets L_{rp}$ 
    \State $unverified\_elements[msg.sender] = []$
    
    
\EndFunction
\algstore{part3}
\end{algorithmic}
\end{algorithm}

The function \textsc{Compute} and some state variables used are defined as follows:

\begin{enumerate}

    \item $V\_curr$ (State Variable): The most recent published value $V$
    \item $V\_curr'$ (State Variable): The most recent published value $V'$
    \item $\theta\_curr$ (State Variable): The most recent published scalar $\theta$
    
\end{enumerate}  

\begin{algorithm}
\caption{Contract $SC_{sys}$: Part 4}\label{SS4}
\begin{algorithmic}[1]
\algrestore{part3}
\Function{Compute}{$V, \theta, V'$}

    \If{$unverified\_elements[msg.sender] \neq []$ OR $block.timestamp < ddl_2$ OR $block.timestamp > ddl_3$}
        \State throw 
    \EndIf    
    
    \If{not $V\_curr$} \Comment{Initialize $V, \theta, V'$}
        \State $V\_curr \gets V$
        \State $\theta\_curr \gets \theta$
        \State $V\_curr' \gets V'$
        \State \Return
    \EndIf    
    
    \State $(rp_{\bm{A}}, rp_{\alpha_{\bm{A}}}, rp_{\alpha_{\bm{A}}\bm{A}}) \gets verified\_elements[msg.sender]$
    \State $acc_1 \gets $\Call{SameRatio}{$(V\_curr, V),rp_{\bm{A}}$}
    \State $acc_2 \gets $\Call{SameRatio}{$(\theta\_curr, \theta), rp_{\alpha_{\bm{A}}}$}
    \State $acc_3 \gets $\Call{SameRatio}{$(V\_curr', V'),rp_{\bm{A}\alpha_{\bm{A}}}$}
    \If {$acc_1 == acc_2 == acc_3 ==$ true }\Comment{Check validity}
        \State $V\_curr \gets V$
        \State $\theta\_curr \gets \theta$
        \State $V\_curr' \gets V'$
    \Else
        \State throw
    \EndIf

\EndFunction
\algstore{part4}
\end{algorithmic}
\end{algorithm}

The function \textsc{Generate} and some state variables used are defined as follows:

\begin{enumerate}
    \item $W\_curr$ (State Variable): The most recent published value $W$
    \item $W\_curr'$ (State Variable): The most recent published value $W'$
    \item $\theta\_curr\_$ (State Variable): The most recent published scalar $\theta$

\end{enumerate}  

\begin{algorithm}
\caption{Contract $SC_{sys}$: Part 5}\label{SS5}
\begin{algorithmic}[1]
\algrestore{part4}
\Function{Generate}{$W, \theta, W'$}

    \If{$unverified\_elements[msg.sender] \neq []$ OR $block.timestamp < ddl_2$ OR $block.timestamp > ddl_3$}
        \State throw 
    \EndIf    
    
    \If{not $W\_curr$} \Comment{Initialize $W, \theta, W'$}
        \State $W\_curr \gets W$
        \State $\theta\_curr\_ \gets \theta$
        \State $W\_curr' \gets W'$
        \State \Return
    \EndIf    
    
    \State $(rp_{\bm{U}}, rp_{\alpha_{\bm{U}}}, rp_{\alpha_{\bm{U}}\bm{U}}) \gets verified\_elements[msg.sender]$
    \State $acc_1 \gets $\Call{SameRatio}{$(W\_curr, W),rp_{\bm{U}}$}
    \State $acc_2 \gets $\Call{SameRatio}{$(\theta\_curr\_, \theta), rp_{\alpha_{\bm{U}}}$}
    \State $acc_3 \gets $\Call{SameRatio}{$(W\_curr', W'),rp_{\alpha_{\bm{U}}\bm{U}}$}
    \If {$acc_1 == acc_2 == acc_3 ==$ true}\Comment{Check validity}
        \State $W\_curr \gets W$
        \State $\theta\_curr\_ \gets \theta$
        \State $W\_curr' \gets W'$
    \Else
        \State throw
    \EndIf

\EndFunction

\end{algorithmic}
\end{algorithm}

\subsection{Authority Setup Contracts} \label{sc-au}
Function \textsc{Commit} and the state variable $h\_collector$ are defined as following:

\begin{enumerate}
    \item $h\_collector$ (State Variable): A mapping collection from the blockchain address belonged to one attribute authority to its commitment $h$
\end{enumerate} 

\begin{algorithm}
\caption{Contract $SC_{auth}$: Part 1}\label{AU1}
\begin{algorithmic}[1]

\Function{Commit}{$h$}
    
    \If {$msg.sender \notin AAlist$ OR $block.timestamp > ddl_1$} 
        \Comment{Basic check}
        \State throw 
    \EndIf

    \If  {$msg.sender \notin h\_collector$} 
        \Comment{Duplicate Sub check}
        \State $h\_collector[msg.sender] \gets h$
    \EndIf

\EndFunction

\algstore{au1}
\end{algorithmic}
\end{algorithm}

The function \textsc{Reveal} and some state variables used are defined as follows:
\begin{enumerate}
    \item $h\_collector$ (State Variable): A mapping collection from the blockchain address belonged to one attribute authority $\bm{AA}$ to its commitment $h'$
    \item $unverified\_sk$ (State Variable): A mapping collection from the blockchain address belonged to one attribute authority $\bm{AA}$ to its list of unverified elements $SK$
    \item $unverified\_elements$ (State Variable): A mapping collection from the blockchain address belonged to one attribute authority $\bm{AA}$ to its list unverified group elements $e'$ in both $\mathbb{G}_1$ and $\mathbb{G}_2$
\end{enumerate}  

\begin{algorithm}
\caption{Contract $SC_{auth}$: Part 2}\label{AU2}
\begin{algorithmic}[1]
\algrestore{au1}
\Function{Reveal}{$L_{sk}, L_{vc}$}
    \If {$msg.sender \notin h\_collector$ OR $block.timestamp > ddl_1'$ OR $msg.sender \in unverified\_elements$} 
        \State throw 
    \EndIf

    \ForAll{$a \in L_{sk}$}
        \State $h \gets h || \Call{Hash}{a}$
    \EndFor

    \ForAll{$(a,b) \in L_{vc}$}
        \State $h \gets g || \Call{Hash}{a} || \Call{Hash}{b}$
    \EndFor
    
    \If {$h\_collector[msg.sender] \neq h$}\Comment{Check validity}
        \State throw 
    \EndIf

    \State $unverified\_elements[msg.sender] \gets L_{vc}$
    \State $unverified\_sk[msg.sender] \gets L_{sk}$
\EndFunction

\algstore{au2}
\end{algorithmic}
\end{algorithm}

The function \textsc{Prove} and related state variables are defined as follows:
\begin{enumerate}
    \item $verified\_elements$ (State Variable): A mapping collection from the blockchain address belonged to one attribute authority to its list of verified group elements in $e'$ 
    \item $verified\_sk$ (State Variable): A mapping collection from the blockchain address belonged to one attribute authority $\bm{AA}$ to its list of verified elements in $PK_{ind}$
    \item $counter$ (State Variable): An integer counter initialized to $0$. It is used to track the number of valid attribute authorities.
    \item $index$ (State Variable): A mapping collection from the blockchain address belonged to one attribute authority to its index $i$
    

\end{enumerate}

\begin{algorithm}
\caption{Contract $SC_{auth}$: Part 3}\label{AU3}
\begin{algorithmic}[1]
\algrestore{au2}
\algblock{Assembly}{End}
\Function{Prove}{$L_{\pi}$}
    
    
    \If{$unverified\_elements[msg.sender] == []$ OR $block.timestamp < ddl_1'$ OR $block.timestamp > ddl_2'$}
        \State throw 
    \EndIf

    \State $L_{e'} \gets unverified_elements[msg.sender]$
    \State $L_{tmp} \gets []$
    \ForAll{$(a,b) \in L_{e'}$}
        \If{not \Call{SameRatio}{$a,b$}}
            \State throw
        \EndIf
        \State $L_{tmp}$.add$(a)$
    \EndFor

    \State $L_{sk} \gets unverified_sk[msg.sender]$
    \State $L_{tmp} \gets L_{tmp} || L_{sk}$
    \ForAll{$i \gets 0,5$}
        \State $rp \gets L_{tmp}[i]$
        \State $h_{tmp} \gets h\_collector[msg.sender] ||$ \Call{Hash}{$rp$}
        \If{not \Call{CheckPoK}{$rp, L_{\pi}[i], h_{tmp}$}}
            \State throw
        \EndIf
    \EndFor

    \State $verified\_elements[msg.sender] \gets L_{e'}$
    \State $unverified\_elements[msg.sender] \gets []$

    \State $verified\_sk[msg.sender] \gets L_{sk}$
    \State $verified\_sk[msg.sender] \gets []$

    \State $counter \gets counter + 1$
    \State $index[msg.sender] \gets counter$
    

\EndFunction
\algstore{au3}
\end{algorithmic}
\end{algorithm}

The function \textsc{Setup} and some state variables used are defined as follows:

\begin{enumerate}

    \item $verified\_O$ (State Variable): A mapping collection from the blockchain address belonged to one attribute authority to its set of elements $(\{o_j^{z_i}\}_{i \neq j, i,j \in [n]}$ in the public parameter of vector commitment $PP_{VC}$

    \item $attribute\_size$ (State Variable): A mapping collection from the blockchain address belonged to one attribute authority to its number of supported attribute

\end{enumerate}

\begin{algorithm}
\caption{Contract $SC_{auth}$}\label{AU4}
\begin{algorithmic}[1]
\algrestore{au3}
\Function{Generate}{$O, \theta, O', l$}
    \If {$msg.sender \notin AAlist$ OR $block.timestamp > ddl$} \Comment{Requirement check}
        \State throw 
    \EndIf
    
    \State $rp_z, rp_{\alpha_z}, rp_{\alpha_z z} \gets verified\_elements[msg.sender]$
    \State $i \gets index[msg.sender]$
    \For {$j \leftarrow 0, n-1$}
        \If{$i == j$}
            \State continue
        \EndIf
        
        \State $check1 \gets \Call{SameRatio}{(rp_z[1], O[j]), rp_z}$
        \State $check2 \gets \Call{SameRatio}{(rp_{\alpha}[1], \theta[j]), rp_{\alpha_z}}$
        \State $check3 \gets \Call{SameRatio}{(rp_{\alpha_z z}[1], O'[j]), rp_{\alpha_z z}}$
        
        \If {$check1 == check2 == check3 ==$ true}
            \State $verified\_O[msg.sender] \gets O$
        \Else
            \State throw
        \EndIf
    \EndFor
    
    \State $attribute\_size[msg.sender] \gets l$
\EndFunction

\end{algorithmic}
\end{algorithm}

\subsection{User Registration Contracts}

Function \textsc{user\_register} and the state variable $collector$ are defined as following:
\begin{enumerate}
    \item $collector$ (State Variable): A mapping collection from the $addr_{DU}$ of $DU$ to the value of global identifier $GID$
\end{enumerate} 

\begin{algorithm}
\caption{Contract $SC_{reg}$}\label{REG}
\begin{algorithmic}[1]

\Function{user\_register}{}
    \If{$msg.value \leq 1000000$} \Comment{Registration fee around 1 USD}
        \State throw
    \EndIf
    
    \State $GID \gets Hash(msg.sender)$
    \State $collector[msg.sender] \gets GID$ \Comment{Assign unique GID}
    
\EndFunction

\end{algorithmic}
\end{algorithm}

\subsection{Utility Function Contracts} \label{sc-util}
The three main functions \textsc{Hash}, \textsc{SameRatio} and \textsc{CheckPoK} are defined as following:
\begin{algorithm}
\caption{Contract $SC_{utils}$: Part 1}\label{util1}
\begin{algorithmic}[1]
\algblock{Assembly}{End}

\Function{Hash}{$m$}
    \State Divide $m$ in small pieces of message blocks $msgs$
    \State $h \gets IV$ \Comment{Pre-defined Initialization Vector}
    \ForAll {$msg \in msgs$}
    
        \State $args \gets $abi.encodePacked$(12, h, msg, \dots)$
    
        \Assembly
            \State $success \gets$ \textbf{staticcall}$(\mathtt{0x09}, args, h)$
            \If{$success == 0$}
                \State revert$(0,0)$
            \EndIf
        \End
    \EndFor
    \State \Return $h$

\EndFunction

\algstore{utilp1}
\end{algorithmic}
\end{algorithm}

\begin{algorithm}
\caption{Contract $SC_{utils}$: Part 2}\label{util2}
\begin{algorithmic}[1]

\algrestore{utilp1}
\algblock{Assembly}{End}

\Function{SameRatio}{$rp1, rp2$}
    \State $input \gets (-rp1[0], rp2[1],rp1[1],rp2[0])$
    \Assembly
        \State $success \gets$ \textbf{staticcall}$(\mathtt{0x08}, input, output)$
        \If{$success == 0$}
            \State revert$(0,0)$
        \EndIf
    \End
    
    \State \Return true
    
\EndFunction
\algstore{utilp2}
\end{algorithmic}
\end{algorithm}

\begin{algorithm}
\caption{Contract $SC_{utils}$: Part 3}\label{util3}
\begin{algorithmic}[1]
\algrestore{utilp2}
\algblock{Assembly}{End}

\Function{CheckPoK}{$rp_s, \pi_s, h_s$}
    
    \State $c \gets $\textsc{Hash}$(\pi_s[0] || h_s)$
    \State $input_1 \gets (rp_s[0], \pi_s[1])$ \Comment{$u \cdot f$}
    \State $input_2 \gets (rp_s[1], c)$ \Comment{$c \cdot H$}
    \Assembly
        \State $success_1 \gets$ \textbf{staticcall}$(\mathtt{0x07}, input_1, result_1)$
        \State $success_2 \gets$ \textbf{staticcall}$(\mathtt{0x07}, input_2, result_2)$
        \If {$success_1 \times success_2 == 0$}
            \State revert$(0,0)$
        \EndIf
        
        \State $input_3 \gets (\pi_s[0], result_2)$ \Comment{$R + c \dot H$}
        \State $success_3 \gets$ \textbf{staticcall}$(\mathtt{0x06}, input_3, result_3)$
        \If {$success_3 == 0$}
            \State revert$(0,0)$
        \EndIf
    \End
    
    \State $output \gets (result_1 == result_3)$
    \State \Return $output$
    
\EndFunction
\end{algorithmic}
\end{algorithm}

\end{document}